\documentclass[authorversion,nonacm,screen]{jair}

\RequirePackage{amsthm}
\usepackage{proof}
\usepackage[location=appendix,manual]{moveproofs}
\usepackage{thm-restate}
\usepackage{centernot}
\usepackage[linesnumbered,noend,boxruled]{algorithm2e}
\usepackage{stmaryrd}
\usepackage{bm}
\usepackage{tikz}
\usetikzlibrary{shapes,calc,arrows,automata}
\usepackage{multirow}
\usepackage{array}
\usepackage{mathtools}
\usepackage{enumitem}
\usepackage{subcaption}
\usepackage{hhline}
\usepackage{stackengine, graphicx}
\usepackage{microtype}
\usepackage{xspace}
\usepackage{wasysym}
\usepackage{listings}
\usepackage{pgfplots}
\usepackage{nicefrac}
\usepackage{pifont}
\usepackage{floatpag}
\floatpagestyle{empty}

\newcommand{\xmark}{\ding{55}}

\usepackage{csquotes}

\RequirePackage[
  datamodel=acmdatamodel,
  style=acmnumeric,
  backend=biber,
  giveninits=true,
  sortcites,
  ]{biblatex}

\newcommand{\nptextcite}[1]{\cite{#1}}

\JAIRAE{Roni Stern}
\JAIRTrack{Conference Award Track}

\AtBeginDocument{%
    \fancyfoot{}
    \fancyfoot{}%

  \fancypagestyle{firstpagestyle}{%
    \fancyhf{}%
    \fancyfoot{}%
  }%
}

\makeatletter
\pretocmd{\thmt@rst@storecounters}{\Hy@SaveLastskip}{}{}
\apptocmd{\thmt@rst@storecounters}{\Hy@RestoreLastskip}{}{}
\makeatother

\makeatletter
\renewcommand{\fnum@figure}{Fig.~\thefigure}
\makeatother

\DontPrintSemicolon

\theoremstyle{definition}
\newtheorem{definition}{Definition}
\newtheorem{example}[definition]{Example}

\theoremstyle{plain}
\newtheorem{theorem}[definition]{Theorem}
\newtheorem{lemma}[definition]{Lemma}

\renewcommand{\epsilon}{\varepsilon}
\let\oldphi\phi
\let\oldvarphi\varphi
\renewcommand{\varphi}{\oldphi}
\renewcommand{\phi}{\oldvarphi}

\newcommand{\inc}{\hspace{-.05em}\raisebox{.4ex}{\tiny\bf ++}}

\newcommand{\VV}{\mathcal{V}}

\newcommand{\ZZ}{\mathbb{Z}}

\newcommand{\NN}{\mathbb{N}}
\newcommand{\RR}{\mathbb{R}}

\newcommand{\TT}{\mathcal{T}}
\newcommand{\LL}{\mathcal{L}}
\renewcommand{\AA}{\mathbf{A}}

\newcommand{\true}{\textup{\textbf{true}}}
\newcommand{\false}{\textup{\textbf{false}}}
\newcommand{\Bool}{\ensuremath{\textbf{Bool}}}
\newcommand{\Int}{\ensuremath{\textbf{Int}}}
\newcommand{\sem}[1]{\left\llbracket #1 \right\rrbracket}
\newcommand{\EIA}{\ensuremath{\text{EIA}}\xspace}
\newcommand{\NRAT}{\ensuremath{\text{NRAT}}\xspace}
\newcommand{\SC}[1]{\textnormal{\scshape #1}}
\let\deriv\partial
\renewcommand{\partial}{\SC{p}}

\newcommand{\NIA}{\ensuremath{\text{NIA}}\xspace}
\newcommand{\LIA}{\ensuremath{\text{LIA}}\xspace}
\newcommand{\sat}{\textup{\textbf{sat}}}
\newcommand{\unsat}{\textup{\textbf{unsat}}}
\newcommand{\unknown}{\textup{\textbf{unknown}}}
\renewcommand{\exp}{\textsf{exp}}
\renewcommand{\mod}{\mathrel{\textsf{mod}}}
\renewcommand{\div}{\mathrel{\textsf{div}}}
\newcommand{\divisible}{\textsf{divisible}}
\newcommand{\tool}[1]{{{\sf{#1}}}}
\newcommand{\swine}{\tool{SwInE}\xspace}
\newcommand{\swinezzz}{\tool{SwInE-Z3}\xspace}
\newcommand{\swinez}{\tool{SwInE-Z3}\xspace}
\newcommand{\swinel}{\tool{SwInE-Legacy}\xspace}

\newcommand{\twoldots}{%
  \mathinner{{\ldotp}{\ldotp}}%
}

\newcommand{\efrac}[2]{%
  \mathchoice
    {\ooalign{%
      $\genfrac{}{}{1.6pt}0{\hphantom{#1}}{\hphantom{#2}}$\cr%
      $\color{white}\genfrac{}{}{1pt}0{\color{black}#1}{\color{black}#2}$}}%
    {\ooalign{%
      $\genfrac{}{}{1.6pt}1{\hphantom{#1}}{\hphantom{#2}}$\cr%
      $\color{white}\genfrac{}{}{1pt}1{\color{black}#1}{\color{black}#2}$}}%
    {\ooalign{%
      $\genfrac{}{}{1.6pt}2{\hphantom{#1}}{\hphantom{#2}}$\cr%
      $\color{white}\genfrac{}{}{1pt}2{\color{black}#1}{\color{black}#2}$}}%
    {\ooalign{%
      $\genfrac{}{}{1.6pt}3{\hphantom{#1}}{\hphantom{#2}}$\cr%
      $\color{white}\genfrac{}{}{1pt}3{\color{black}#1}{\color{black}#2}$}}%
}

\newcommand{\Def}{\mathrel{\mathop:}=}
\newcommand{\mDo}{\mathbf{do}}

\newcommand{\mDone}{\mathbf{done}}
\newcommand{\mWhile}[2]{\mathbf{while}\ #1\ \mDo\ #2\ \mDone}

\newcommand{\assign}{\leftarrow}
\SetKwRepeat{Do}{do}{while}
\SetKwFor{While}{while}{}{end}%
\SetKwFor{For}{for}{}{end}%

\newcommand{\ip}[2]{\mathrm{ip}_{#1}^{[{#2}^\pm]}} 
\newcommand{\ipc}[3]{\mathrm{ip}_{#1}^{[{#2} \twoldots {#3}]}} %for concrete values
\newcommand{\ipf}[2]{\mathrm{ip}^{[{#1}^\pm][{#2}^\pm]}}
\newcommand{\ipfc}[4]{\mathrm{ip}^{[{#1} \twoldots {#2}][{#3} \twoldots {#4}]}}

\newcommand{\omitproof}[1]{}

\renewcommand{\emptyset}{\varnothing}

\crefname{algorithm}{alg.}{algorithms}%
\crefname{equation}{eq.}{equations}%
\crefname{chapter}{chapter}{chapters}%
\crefname{section}{sect.}{sections}%
\crefname{appendix}{app.}{appendices}%
\crefname{enumi}{item}{items}%
\crefname{footnote}{footnote}{footnotes}%
\crefname{figure}{fig.}{figures}%
\crefname{table}{table}{tables}%
\crefname{theorem}{thm.}{theorems}%
\crefname{lemma}{lemma}{lemmas}%
\crefname{corollary}{cor.}{corollaries}%
\crefname{proposition}{proposition}{propositions}%
\crefname{definition}{def.}{definitions}%
\crefname{result}{result}{results}%
\crefname{example}{ex.}{examples}%
\crefname{remark}{remark}{remarks}%
\crefname{note}{note}{notes}%
\crefname{invariant}{invariant}{invariants}%

\begin{document}

\title{Satisfiability Modulo Exponential Integer Arithmetic}
\author{Florian Frohn}
\orcid{http://orcid.org/0000-0003-0902-1994}
\email{florian.frohn@informatik.rwth-aachen.de}
\affiliation{
  \institution{RWTH Aachen University}
  \city{Aachen}
  \country{Germany}
}
\author{Jürgen Giesl}
\orcid{http://orcid.org/0000-0003-0283-8520}
\email{giesl@informatik.rwth-aachen.de}
\affiliation{
  \institution{RWTH Aachen University}
  \city{Aachen}
  \country{Germany}
}

\renewcommand{\shortauthors}{F.\ Frohn, J.\ Giesl}

\begin{abstract}
  {\bf Background:}
  SMT solvers use sophisticated techniques for polynomial (linear or non-linear) integer arithmetic.
  In contrast, non-polynomial integer arithmetic has mostly been neglected so far.
  However, in the context of program verification, polynomials are often insufficient to capture the behavior of the analyzed system without resorting to approximations.
  In particular, exponentials are often required to facilitate precise reasoning (without approximations).

  {\bf Objectives:}
  We introduce \EIA, an SMT theory for integer arithmetic with exponentiation, as well as a technique to analyze satisfiability of \EIA problems.
  Moreover, we provide a set of \EIA benchmarks and an implementation of our approach.

  {\bf Methods:}
  In the last years, \emph{incremental linearization} has been applied successfully to satisfiability modulo real arithmetic with transcendental functions.
  We adapt this approach to an extension of polynomial integer arithmetic with exponential functions.
  Here, the key challenge is to compute suitable \emph{lemmas} that eliminate the current model from the search space if it violates the semantics of exponentiation.
  We implemented our approach in the novel tool \swine.

  {\bf Results:}
  We evaluate our approach empirically on several sets of benchmarks from different domains:
  Most of them stem from program verification, and one set stems from recurrence solving.
  On all sets, our approach is highly effective in practice.

  {\bf Conclusions:}
  We develop a novel approach for  satisfiability modulo exponential integer
  arithmetic.
  With \swine, we provide an SMT-LIB compliant open-source solver for \EIA
that demonstrates that our approach is effective in practice. 
  In this way, we hope to attract users with applications that give rise to challenging benchmarks, and we hope that other solvers with support for integer exponentiation will follow, with the ultimate goal of standardizing \EIA.

\end{abstract}

\maketitle

\section{Introduction}
\label{sec:intro}

Traditionally, automated reasoning techniques for integers focus on polynomial arithmetic.
This is not only true in the context of SMT, but also for program verification techniques, since the latter often search for polynomial invariants that imply the desired properties.
As invariants are over-approximations, they are well suited for proving ``universal'' properties like safety, termination, or upper bounds on the worst-case runtime that refer to all possible program runs.
However, proving dual properties like unsafety, non-termination, or lower bounds requires under-approximations, so that invariants are of limited use here.

For lower bounds, an \emph{infinite set} of witnesses is needed, as the runtime w.r.t.\ a finite set of (terminating) program runs is always bounded by a constant.
Thus, to prove non-constant lower bounds, \emph{symbolic under-approximations} are required, i.e., formulas that describe an infinite subset of the reachable states.
However, polynomial arithmetic is often insufficient to express such approximations.
To see this, consider the program
\[
	x \assign 1;\ y \assign \texttt{nondet}(0,\infty);\ \mWhile{y > 0}{x \assign 3 \cdot x;\ y \assign y - 1}
\]
where $\texttt{nondet}(0,\infty)$ returns a natural number non-deterministically.
Here, the set of reachable states after execution of the loop is characterized by the formula
\begin{equation}
	\label{eq:exp}
	\exists n \in \NN.\ x = 3^n \land y = 0.
\end{equation}
In recent work, \emph{acceleration techniques} have successfully been used to deduce lower runtime bounds automatically \cite{loat,journal}.
While they can easily derive a formula like \eqref{eq:exp} from the code above, this is of limited use, as most\footnote{\tool{CVC5} uses a dedicated solver for \href{https://github.com/cvc5/cvc5/blob/a287aabdadbead4b5bbebaa4c57818cc4c3f207e/src/theory/arith/nl/pow2_solver.h}{integer exponentiation with base $2$}.}
SMT solvers cannot handle terms of the form $3^n$.
Besides lower bounds, acceleration has also successfully been used for proving non-termination \cite{adcl-nt,fmcad19,loat} and (un)safety \cite{flata,underapprox15,fast,bozga10,jeannet14,adcl}, where its strength is finding long counterexamples that are challenging for other techniques.

Importantly, exponentiation is not just ``yet another function'' that can result from applying acceleration techniques.
There are well-known, important classes of loops where polynomials and exponentiation \emph{always} suffice to represent the values of the program variables after executing a loop \cite{FMSD23,cav19}.
Thus, the lack of support for integer exponentiation in SMT solvers is a major obstacle for the further development of acceleration-based verification techniques.

In this work, we first define a novel SMT theory for integer arithmetic with exponentiation.
Then we show how to lift standard SMT solvers to this new theory, resulting in our novel tool \swine (\underline{S}MT \underline{w}ith \underline{In}teger \underline{E}xponentiation).

Our technique is inspired by \emph{incremental linearization}, which has been applied successfully to \emph{real}
arithmetic with transcendental functions, including the natural exponential function $\exp_e(x) = e^x$, where $e$ is Euler's number \cite{cimatti18a}.
In this setting, incremental linearization considers $\exp_e$ as an uninterpreted function.
If the resulting SMT problem is unsatisfiable, then so is the original problem.
If it is satisfiable and the model that was found for $\exp_e$ coincides with the semantics of exponentiation, then the original problem is satisfiable.
Otherwise, \emph{lemmas} about $\exp_e$ that rule out the current model are added to the SMT problem, and then its satisfiability is checked again.
The name ``incremental linearization'' is due to the fact that these lemmas only contain linear arithmetic.

The main challenge for adapting this approach to integer exponentiation is to generate suitable lemmas, see \Cref{sec:linearization}.
Except for \emph{monotonicity lemmas}, none of the lemmas of \citeauthor{cimatti18a} easily carry over to our setting.
In contrast to \citeauthor{cimatti18a}, we do not restrict ourselves to linear lemmas, but we also use non-linear, polynomial lemmas.
This is due to the fact that we consider a binary version $\lambda x, y.\ x^y$ of exponentiation, whereas \citeauthor{cimatti18a} fix the base to $e$.
Thus, in our setting, one obtains \emph{bilinear} lemmas that are linear w.r.t.\ $x$ as well as $y$, but may contain multiplication between $x$ and $y$ (i.e., they may contain the subterm $x \cdot y$).
More precisely, bilinear lemmas arise from \emph{bilinear interpolation}, which is a crucial ingredient of our approach, as it allows us to eliminate \emph{any} model that violates the semantics of exponentiation (\Cref{ProgressTheorem}).
Therefore, the name ``incremental linearization'' does not fit to our approach, which is rather an instance of ``counterexample-guided abstraction refinement'' (CEGAR, \nptextcite{Clarke2000}).

A preliminary shorter version of this paper appeared at IJCAR 2024 \cite{conference}.
The current paper contains the following, yet unpublished contributions:
\begin{itemize}
\item We present two new kinds of lemmas, called \emph{prime lemmas} and \emph{induction lemmas}.
\item We present a new technique called \emph{phasing}, which alternates between a \emph{\sat-phase} where the solver searches for a model within a restricted part of the search space, and an \emph{\unsat-phase} where the solver aims for proving unsatisfiability.
\item While
 \textcite{conference} 
  did not give any proofs, we now present proofs for all lemmas and theorems.
\item We present a reimplementation of \swine, called \swinezzz in the sequel, which directly builds on top of the API of the SMT solver \tool{Z3} \cite{z3}.
  In contrast, the original implementation of  \textcite{conference} 
  was built on top of \tool{SMT-Switch} \cite{smt-switch}, a library that offers a unified interface for various SMT solvers.
  However, the evaluation of \textcite{conference} showed that \tool{Z3} performs best in our setting, so that using its API directly allows for improving robustness.\footnote{For example, \swine uses the parser of \tool{SMT-Switch}, which strictly adheres to the SMT-LIB standard.
  In contrast, \tool{Z3} (and thus \swinezzz) can often also parse inputs that violate the SMT-LIB standard in a way that does not introduce ambiguity.}
\item We present a novel collection of challenging unsatisfiable benchmarks, which were obtained from verifying solutions for recurrence relations.
\item We improve and extend our evaluation by including our novel benchmarks and the tools \swinezzz and \tool{Z3}.
  In order to use \tool{Z3} for integer exponentiation, we axiomatize the semantics of exponentiation via universally quantified assertions.
\end{itemize}
To summarize, our contributions are as follows: We first propose the new SMT theory $\EIA$ for integer arithmetic with exponentiation (\Cref{sec:theories}).
Then, based on novel techniques for generating suitable lemmas, we develop a CEGAR approach for \EIA (\Cref{sec:solving}).
Next, we show how to improve our approach via phasing (\Cref{sec:phasing}).
After discussing related work (\Cref{sec:related}), we discuss two open-source implementations of our approach, one on top of
\tool{SMT-Switch} \cite{swine-github}, and one on top of \tool{Z3} \cite{swine-z3-github}, and we
present an experimental evaluation on several collections of \EIA benchmarks that we synthesized from
verification problems and from recurrence solving.
Our experiments show that our approach is highly effective in practice (\Cref{sec:evaluation}).
We refer to
App.\ \ref{sec:MissingProofs} for those (parts of the) proofs that were omitted
from the main part of the paper.

\section{Preliminaries}
\label{sec:preliminaries}

We are working in the setting of \emph{SMT-LIB logic} \cite{smtlib}, a variant of many-sorted first-order logic with equality.
We now introduce a reduced variant of SMT-LIB logic, where we only explain those concepts that are relevant for our work.

In SMT-LIB logic, there is a dedicated Boolean sort \Bool, and hence formulas are just terms of sort \Bool.
Similarly, there is no distinction between predicates and functions, as predicates are
simply functions with results of sort \Bool. 

So in SMT-LIB logic, a \emph{signature} $\Sigma = (\Sigma^S,\Sigma^F,\Sigma^R)$ consists of a set $\Sigma^S$ of \emph{sorts}, a set $\Sigma^F$ of \emph{function symbols}, and a \emph{ranking function} $\Sigma^R: \Sigma^F \to (\Sigma^S)^+$.
The meaning of $\Sigma^R(f) = (s_1,\ldots,s_k)$ is that $f$ is a function which maps arguments of the sorts $s_1,\ldots,s_{k-1}$ to a result of sort $s_k$.
We write $f: s_1\ \ldots\ s_k$ instead of ``$f \in \Sigma^F$ and $\Sigma^R(f) = (s_1,\ldots,s_k)$'' if $\Sigma$ is clear from the context.
We always allow to implicitly extend $\Sigma$ with arbitrarily many constant function symbols (i.e., function symbols $x$ where $|\Sigma^R(x)| = 1$).
Note that SMT-LIB logic only considers closed terms, i.e., terms without free variables, and we are only concerned with quantifier-free formulas, so in our setting, all formulas are ground.
Therefore, we refer to these constant function symbols as \emph{variables} to avoid
confusion with other, predefined constant function symbols like $\true, 0, 1, \ldots$, see below.

Every SMT-LIB signature is an extension of $\Sigma_\Bool$ where $\Sigma^S_\Bool = \{\Bool\}$ and $\Sigma^F_\Bool$ consists of the following function symbols:
\[
	\true, \false: \Bool \qquad \neg: \Bool\ \Bool \qquad \land, \lor, \implies, \iff: \Bool\ \Bool\ \Bool
\]
Note that SMT-LIB logic only considers well-sorted terms.
A \emph{$\Sigma$-structure} $\AA$ consists of a \emph{universe} $A = \bigcup_{s \in \Sigma^S} A_s$ and an \emph{interpretation function} that maps each function symbol $f: s_1\ \ldots\ s_k$ to a function $\sem{f}^{\AA}: A_{s_1} \times \ldots \times A_{s_{k-1}}
	\to A_{s_k}$.
SMT-LIB logic only considers structures where $A_\Bool = \{\true,\false\}$ and all function symbols from $\Sigma_\Bool$ are interpreted as usual.

A \emph{$\Sigma$-theory} is a class of $\Sigma$-structures.
For example, consider the extension $\Sigma_\Int$ of $\Sigma_\Bool$ with the additional sort $\Int$ and the following function symbols:
\begin{align*}
  c: &\ \Int && \text{for all } c \in \NN\\
  +,-,\cdot,\div,\mod:&\ \Int\ \Int\ \Int\\
  <, \leq, >, \geq, =, \neq:&\ \Int\ \Int\ \Bool\\
  \divisible_c:&\ \Int\ \Bool && \text{for all } c \in \NN_{>0} = \{n \in \NN \mid n > 0\}
\end{align*}
Then the $\Sigma_\Int$-theory \emph{non-linear integer arithmetic} (\NIA)\footnote{As we only consider quantifier-free formulas, we omit the prefix ``QF\_'' in theory names and write, e.g., \NIA instead of QF\_NIA.
  \textcite{smtlib} call QF\_NIA an \emph{SMT-LIB logic}, which restricts the (first-order) \emph{theory} of integer arithmetic to the quantifier-free fragment.
  For simplicity, we do not distinguish between SMT-LIB logics and theories.} contains all $\Sigma_\Int$-structures $\AA$ where $A_\Int = \ZZ$,
\[
  \sem{\divisible_c}^\AA(x) =
  \begin{cases}
    \true & \text{if } x \mod c = 0 \\
    \false & \text{otherwise},
  \end{cases}
\]
and all other symbols from $\Sigma_\Int$ are interpreted as usual.\footnote{Note that all functions are total in SMT-LIB logic.
  So $\div$ and $\mod$ are interpreted as integer division and modulo, respectively, if their second argument is non-zero.
  If the second argument is zero, then the interpretation of $\div$ and $\mod$ is arbitrary, i.e., $\lambda x.\ x \div 0$ and $\lambda x.\ x \mod 0$ are treated like un\-interpreted functions.}

If $\AA$ is a $\Sigma$-structure and $\Sigma'$ is a subsignature of $\Sigma$, then the \emph{reduct} of $\AA$ to $\Sigma'$ is the unique $\Sigma'$-structure that interprets its function symbols like $\AA$.
So the theory \emph{linear integer arithmetic} (\LIA) consists of the reducts of all
elements of \NIA to $\Sigma_\LIA = \Sigma_\Int \setminus \{\cdot,\div,\mod\}$.\footnote{Note that $\Sigma_\LIA$ indeed contains the predicates $\divisible_c$, as \LIA would not admit quantifier elimination without them.
  To see this, consider the formula $\phi \Def \exists y.\ x = y + y$.
  It is equivalent to $\divisible_2(x)$, but there is no quantifier-free \LIA formula without divisibility predicates which is equivalent to $\phi$.}

Given a $\Sigma$-structure $\AA$ and a $\Sigma$-term $t$, the \emph{meaning} $\sem{t}^\AA$ of $t$ results from interpreting all function symbols according to $\AA$.
For function symbols $f$ whose interpretation is fixed by a $\Sigma$-theory $\TT$, we denote $f$'s interpretation by $\sem{f}^\TT$.
Given a $\Sigma$-theory $\TT$, a $\Sigma$-formula $\phi$ (i.e., a $\Sigma$-term of sort $\Bool$) is \emph{satisfiable in $\TT$} if there is an $\AA \in \TT$ such that $\sem{\phi}^\AA = \true$.
Then $\AA$ is called a \emph{model} of $\phi$, written $\AA \models \phi$.
If \emph{every} $\AA \in \TT$ is a model of $\phi$, then $\phi$ is \emph{$\TT$-valid}, written $\models_\TT \phi$.
We write $\psi \equiv_\TT \phi$ for $\models_\TT \psi \iff \phi$.

We sometimes also consider \emph{uninterpreted functions}.
Then the signature may not only contain the function symbols of the theory under consideration and variables, but also additional non-constant function symbols.

We write ``term'', ``structure'', ``theory'', $\ldots$ instead of ``$\Sigma$-term'', ``$\Sigma$-struc\-ture'', ``$\Sigma$-theory'', $\ldots$ if $\Sigma$ is irrelevant or clear from the context.
Similarly, we just write ``$\equiv$'' and ``valid'' instead of ``$\equiv_\TT$'' and ``$\TT$-valid'' if $\TT$ is clear from the context.
Moreover, we use unary minus and $t^c$ (where $t$ is a term of sort $\Int$ and $c \in
\NN$) as syntactic sugar, and we write binary function symbols in
infix notation.

In the sequel, we use $x,y,z,\ldots$ for variables, $s,t,p,q,\ldots$ for terms of sort $\Int$, $\phi, \psi,\ldots$ for formulas, and $a,b,c,d, \ldots$ for integers.

\section{The SMT Theory \EIA}
\label{sec:theories}

We now introduce our novel SMT theory for \emph{exponential integer arithmetic}.
To this end, we define the signature $\Sigma_\Int^{\exp}$, which extends $\Sigma_\Int$ with
\[
	\exp: \Int\ \Int\ \Int.
\]
If the $2^{nd}$ argument of $\exp$ is non-negative, then its semantics is as expected, i.e., we are interested in structures $\AA$ such that $\sem{\exp}^\AA(c,d) = c^d$ for all $d \geq 0$.
However, if the $2^{nd}$ argument is negative, then we have to use different semantics.
The reason is that we may have $c^d \notin \ZZ$ if $d < 0$.
Intuitively, $\exp$ should be a partial function, but all functions are total in SMT-LIB logic.
We solve this problem by interpreting $\exp(c,d)$ as $c^{|d|}$.
This semantics has previously been used in the literature, and the resulting logic admits a known decidable fragment \cite{LIAE}.

\begin{definition}[\EIA]
	\label{EIAT}
	The theory \emph{exponential integer arithmetic (\EIA)}
	contains all \normalfont{$\Sigma_\Int^{\exp}$}-structures $\AA$ with
        \[
          \normalfont{\sem{\exp}^\AA(c,d) = c^{|d|}}
        \]
          whose reduct to \normalfont{$\Sigma_\Int$} is in $\NIA$.
\end{definition}
Alternatively, one could treat $\exp(c,d)$ like an uninterpreted function if $d$ is negative.
Doing so would be analogous to the treatment of division by zero in SMT-LIB logic.
Then, e.g., $\exp(0,-1) \neq \exp(0,-2)$ would be satisfied by a structure $\AA$ with $\sem{\exp}^\AA(c,d) = c^d$ if $d \geq 0$ and $\sem{\exp}^\AA(c,d) = d$, otherwise.
However, the drawback of this approach is that important laws of exponentiation like
\[
	\exp(\exp(x,y), z) = \exp(x,y \cdot z)
\]
would not be valid.
Thus, we focus on the semantics from \Cref{EIAT}.

\section{Solving \EIA Problems via CEGAR}
\label{sec:solving}

\begin{algorithm}[t]
	\caption{CEGAR for $\EIA$}\label{alg}
	\KwIn{a $\Sigma_\Int^\exp$-formula $\phi$}
	\tcp{Preprocessing}
	\Do{$\phi \neq \phi'$\label{preprocess}}{
		$\phi' \gets \phi$\;
		$\phi \gets \SC{FoldConstants}(\phi)$\;\label{fold}
		$\phi \gets \SC{Rewrite}(\phi)$\;\label{rewrite}
	}
	\tcp{Refinement Loop}
	\While{there is a \NIA-model $\AA$ of $\phi$ \label{counterexample}}{
		\If{$\AA$ is a counterexample}{
		  $\LL \gets \emptyset$\;\label{refine0}
			\For{$\mathrm{k}\!\in\!\{ \text{Symmetry}, \text{Monotonicity}, \text{Bounding}, \text{Prime}, \text{Induction}, \text{Interpolation}\}$ \label{refine1}}{
				$\LL \assign \LL \cup \SC{ComputeLemmas}(\phi, \mathrm{k})$\;
			}
			$\phi \gets \phi \land \bigwedge \{\psi \in \LL \mid \AA \not\models \psi\}$\label{refine2}
		}
		\lElse{
			\Return{$\sat$}
		}
	}
	\Return{$\unsat$}
\end{algorithm}
We now explain our technique for solving \EIA problems, see Alg.~\ref{alg}.
Our goal is to (dis)prove satisfiability of $\phi$ in $\EIA$.
The loop in Line \ref{counterexample} is a CEGAR loop which lifts an SMT solver for \NIA (which is called in Line \ref{counterexample}) to \EIA.
So the \emph{abstraction} consists of using \NIA- instead of \EIA-models.
Hence, $\exp$ is considered to be an uninterpreted function in Line~\ref{counterexample}, i.e., the SMT solver also searches for an interpretation of $\exp$.
If the model found by the SMT solver is a \emph{counterexample}
(i.e., if $\sem{\exp}^\AA$ conflicts with $\sem{\exp}^{\EIA}$), then the formula under consideration is refined by adding suitable lemmas in Lines \ref{refine1} -- \ref{refine2} and the loop is iterated again.
\begin{definition}[Counterexample]
	\label{def:counterexample}
	We call a \NIA-model $\AA$ of $\phi$ a \emph{counterexample} if there is a subterm \normalfont{$\exp(s,t)$} of $\phi$ such that \normalfont{$\sem{\exp(s,t)}^\AA \neq (\sem{s}^{\AA})^{|\sem{t}^{\AA}|}$}.
\end{definition}
In the sequel, we first discuss our preprocessings (first loop in Alg.~\ref{alg}) in \Cref{sec:preprocessings}.
Then we explain our refinement (Lines \ref{refine0} -- \ref{refine2}) in \Cref{sec:linearization}.
Here, we first introduce the different kinds of lemmas that are used by our implementation in \Cref{subsec:sym} -- \ref{sec:interpolation}.
If implemented naively, the number of lemmas can get quite large, so we explain how to generate lemmas \emph{lazily} in \Cref{subsec:lazy}.
Finally, we conclude this section by stating important properties of Alg.~\ref{alg}.
\begin{example}[Leading Example]
	\label{ex:leading}
	To illustrate our approach, we show how to prove
	\begin{equation}
          \label{eq:leading}
		\forall x,y.\ |x| > 2 \land |y| > 2 \implies \exp(\exp(x,y),y) \neq \exp(x,\exp(y,y)).
	\end{equation}
        Here, we also consider negative values for $x$ and $y$ (i.e., we use the premise $|x| > 2 \land |y| > 2$ instead of $x > 2 \land y > 2$) to illustrate how our techniques works when $\exp$ is used with negative arguments.
	To prove \eqref{eq:leading}, we encode absolute values suitably\footnote{We tested several encodings, but surprisingly, this non-linear encoding worked best.} and prove unsatisfiability of its negation:
	\[
		x^2 > 4 \land y^2 > 4 \land \exp(\exp(x,y),y) = \exp(x,\exp(y,y))
	\]
\end{example}

\subsection{Preprocessings}
\label{sec:preprocessings}

In the first loop of Alg.~\ref{alg}, we preprocess $\phi$ by alternating \emph{constant folding} (Line \ref{fold}) and \emph{rewriting} (Line \ref{rewrite}) until a fixpoint is reached.
Constant folding evaluates subexpressions without variables, where subexpressions $\exp(c,d)$ are evaluated to $c^{|d|}$, i.e., according to the semantics of \EIA.
Rewriting reduces the number of occurrences of $\exp$ via the following (terminating) rewrite rules:
\begin{align*}
	\exp(x,c)                 & {} \to x^{|c|} \qquad\qquad\qquad \text{if $c \in \ZZ$} \\
	\exp(\exp(x,y),z)         & {} \to \exp(x, y \cdot z)                               \\
	\exp(x,y) \cdot \exp(z,y) & {} \to \exp(x \cdot z, y)
\end{align*}
In particular, the $1^{st}$ rule allows us to rewrite\footnote{Note that we have $\sem{\exp(0,0)}^\EIA = 0^0 = 1$.} $\exp(s,0)$ to $s^0 = 1$ and $\exp(s,1)$ to $s^1 = s$.
Note that the rule
\[
	\exp(x,y) \cdot \exp(x,z) \to \exp(x,y+z)
\]
would be unsound, as the right-hand side would need to be $\exp(x,|y|+|z|)$ instead.
\begin{example}[Preprocessing]
	\label{ex:preprocessing}
	For our leading example, applying the $2^{nd}$ rewrite rule at the underlined position yields:
	\begin{align}
		          & x^2 > 4 \land y^2 > 4 \land \underline{\exp(\exp(x,y),y)} = \exp(x,\exp(y,y)) \nonumber \\
		{} \to {} & x^2 > 4 \land y^2 > 4 \land \exp(x,y^2) = \exp(x,\exp(y,y)) \label{eq:goal}
	\end{align}
\end{example}
\begin{restatable}{lemma}{preprocsound}
	\label{lem:preproc}
	We have $\phi \equiv_{\normalfont{\EIA}} \SC{FoldConstants}(\phi)$ and $\phi \equiv_{\normalfont{\EIA}} \SC{Rewrite}(\phi)$.
\end{restatable}
\begin{proof}
  For constant folding, the claim is trivial, so we only have to show $\phi \equiv_{\normalfont{\EIA}} \SC{Rewrite}(\phi)$.
  To this end, it suffices to show $\sem{\ell}^\AA = \sem{r}^\AA$ for all rewrite rules $\ell \to r$ and all $\AA \in \EIA$.
  Let $\AA \in \EIA$ be arbitrary but fixed and let $\sem{\ldots}$ denote $\sem{\ldots}^\AA$.

  For the first rewrite rule $\exp(x,c) \to x^{|c|}$, recall that $x^{|c|}$ is syntactic sugar for $\overbrace{x \cdot \ldots \cdot x}^{|c| \text{ times}}$, i.e., we have to show $\sem{\exp(x,c)} = \llbracket \overbrace{x \cdot \ldots \cdot x}^{|c| \text{ times}} \rrbracket$.
  We obtain
  \[
    \sem{\exp(x,c)} = \sem{x}^{|\sem{c}|} = \sem{x}^{|c|} = \overbrace{\sem{x} \cdot \ldots \cdot \sem{x}}^{|c| \text{ times}} = \llbracket \overbrace{x \cdot \ldots \cdot x}^{|c| \text{ times}} \rrbracket.
  \]
  The proofs for the remaining rewrite rules are analogous.
  See App.\ \ref{sec:MissingProofs} for the complete proofs.
\end{proof}
\makeproof*{lem:preproc}{
  \preprocsound*
	\begin{proof}
	\omitproof{Constant folding is trivially sound, so we only have to show soundness of our rewrite rules.
		To this end, it suffices to show $\sem{\ell}^\AA = \sem{r}^\AA$ for all rewrite rules $\ell \to r$ and all $\AA \in \EIA$.
		Let $\AA \in \EIA$ be arbitrary but fixed and let $\sem{\ldots}$ denote $\sem{\ldots}^\AA$.

		For the first rewrite rule $\exp(x,c) \to x^{|c|}$, recall that $x^{|c|}$ is syntactic sugar for $\overbrace{x \cdot \ldots \cdot x}^{|c| \text{ times}}$, i.e., we have to show $\sem{\exp(x,c)} = \llbracket \overbrace{x \cdot \ldots \cdot x}^{|c| \text{ times}} \rrbracket$.
		We obtain
		\[
			\sem{\exp(x,c)} = \sem{x}^{|\sem{c}|} = \sem{x}^{|c|} = \overbrace{\sem{x} \cdot \ldots \cdot \sem{x}}^{|c| \text{ times}} = \llbracket \overbrace{x \cdot \ldots \cdot x}^{|c| \text{ times}} \rrbracket.
		\]
        }
               In \Cref{sec:solving}, we already showed the soundness of the first rewrite rule.
		For the second rewrite rule $\exp(\exp(x,y),z) \to \exp(x,y \cdot z)$, we have
                \[
		\begin{array}{rcccccl}
			\sem{\exp(\exp(x,y),z)} & = & \sem{\exp(x,y)}^{|\sem{z}|} & = & \left(\sem{x}^{|\sem{y}|}\right)^{|\sem{z}|} & = &\sem{x}^{|\sem{y}| \cdot |\sem{z}|} \\
			                        & = & \sem{x}^{|\sem{y} \cdot \sem{z}|} & = & \sem{x}^{|\sem{y \cdot z}|} & = & \sem{\exp(x,y \cdot z)}.
		\end{array}
                \]

		For the third rewrite rule $\exp(x,y) \cdot \exp(z,y) \to \exp(x \cdot z, y)$, we have
                \[
		\begin{array}{rcccccl}
			\sem{\exp(x,y) \cdot \exp(z,y)} & = & \sem{\exp(x,y)} \cdot \sem{\exp(z,y)} & = & \sem{x}^{|\sem{y}|} \cdot \sem{z}^{|\sem{y}|} & = & \left(\sem{x} \cdot \sem{z}\right)^{|\sem{y}|} \\
			                                & = & \sem{x \cdot z}^{|\sem{y}|} & = & \sem{\exp(x \cdot z, y)}.
		\end{array}
                \]
	\end{proof}
}

\subsection{Refinement}
\label{sec:linearization}
Our refinement (Lines \ref{refine0} -- \ref{refine2} of Alg.~\ref{alg}) is based on the six kinds of lemmas named in Line \ref{refine1}: \emph{symmetry lemmas}, \emph{monotonicity lemmas}, \emph{bounding lemmas}, \emph{prime lemmas}, \emph{induction lemmas}, and \emph{interpolation lemmas}.
In the sequel, we explain how we compute a set $\LL$ of such lemmas.
Then our refinement conjoins
\[
	\{\psi \in \LL \mid \AA \not \models \psi\}
\]
to $\phi$ in Line \ref{refine2}.
As our lemmas allow us to eliminate \emph{any} counterexample, this set is never empty, see \Cref{ProgressTheorem}.
To compute $\LL$, we consider all terms that are \emph{relevant} for the formula $\phi$.
\begin{definition}[Relevant Terms]
	\label{def:relevant}
	A term $\normalfont{\exp}(s,t)$ is \emph{relevant} if $\phi$ has a subterm of the form $\normalfont{\exp}(\pm s,\pm t)$.
\end{definition}
\begin{example}[Relevant Terms]
	For our leading example \eqref{eq:goal}, the relevant terms are all terms of the form $\exp(\pm x, \pm y^2)$, $\exp(\pm y, \pm y)$, or $\exp(\pm x, \pm \exp(y,y))$.
\end{example}
While the formula $\phi$ is changed in Line \ref{refine2} of Alg.~\ref{alg}, we only conjoin new lemmas to $\phi$, and thus relevant terms can never become irrelevant.
Moreover, by construction our lemmas only contain $\exp$-terms that were already relevant before.
Thus, the set of relevant terms is not changed by our CEGAR loop.

As mentioned in \Cref{sec:intro}, our approach may also compute lemmas with non-linear polynomial arithmetic.
However, our lemmas are linear if $s$ is an integer constant and $t$ is linear for all subterms $\exp(s,t)$ of $\phi$.

\subsubsection{Symmetry Lemmas}
\label{subsec:sym}
\emph{Symmetry lemmas} encode the relation between terms of the form $\exp(\pm s,\pm t)$.
For each relevant term $\exp(s,t)$, the set $\LL$ contains the following symmetry lemmas:
\begin{align}
	\divisible_2(t) \implies {} & \exp(s,t) = \exp(-s,t) \label{lem:symmetry-t-1}
	\tag{\protect{\ensuremath{\SC{sym}_1}}}                                     \\
	\neg\divisible_2(t) \implies {} & \exp(s,t) = -\exp(-s,t) \label{lem:symmetry-t-2}
	\tag{\protect{\ensuremath{\SC{sym}_2}}}                                     \\
	                         & \exp(s,t) = \exp(s,-t) \label{lem:symmetry-t-3}
	\tag{\protect{\ensuremath{\SC{sym}_3}}}
\end{align}
Note that \ref{lem:symmetry-t-1} and \ref{lem:symmetry-t-2}
are just implications, not equivalences, as, for example, $c^{|d|} = (-c)^{|d|}$ does not imply $\divisible_2(d)$ if $c = 0$.
\begin{example}[Symmetry Lemmas]
	\label{ex:symmetry}
	For our leading example \eqref{eq:goal}, the following symmetry lemmas would be considered, among others:
	\begin{align}
		\ref{lem:symmetry-t-1}: &  & \divisible_2(-y) \implies {} & \exp(-y,-y) = \exp(y,-y) \label{ex:sym_a}
		\\
		\ref{lem:symmetry-t-2}: &  & \neg\divisible_2(-y) \implies {} & \exp(-y,-y) = -\exp(y,-y) \label{ex:sym_b}
		\\
		\ref{lem:symmetry-t-3}: &  &                           & \exp(x, \exp(y,y)) = \exp(x, -\exp(y,y)) \label{ex:sym_c}
		\\
		\ref{lem:symmetry-t-3}: &  &                           & \exp(y,y) = \exp(y,-y) \label{ex:sym_d}
	\end{align}
	Note that, e.g., \eqref{ex:sym_a} results from the term $\exp(-y,-y)$, which is relevant (see \Cref{def:relevant}) even though it does not occur in $\phi$.
\end{example}
To prove soundness of our refinement, we have to show that our lemmas are $\EIA$-valid.
\begin{restatable}{lemma}{symmetrysound}
	\label{lem:symmetry}
	Let $s,t$ be terms of sort \normalfont{$\Int$}.
	Then \ref{lem:symmetry-t-1} -- \ref{lem:symmetry-t-3} are \EIA-valid.
\end{restatable}
\begin{proof}
  Let $\AA \in \EIA$ again be arbitrary but fixed and let $\sem{\ldots}$ denote $\sem{\ldots}^\AA$.
  For \ref{lem:symmetry-t-1}, assume $\sem{\divisible_2(t)} = \true$, i.e., assume that $\sem{t}$ is even.
  Then it remains to show
  \[
    \sem{\exp(s,t) = \exp(-s,t)} = \true, \quad \text{i.e.,} \quad \sem{\exp(s,t)} = \sem{\exp(-s,t)}.
  \]
  We have:
  \begin{align*}
    \sem{\exp(s,t)} & {} = \sem{s}^{|\sem{t}|}                                                         \\
                    & {} = \sem{-s}^{|\sem{t}|} \tag{as $\sem{t}$, and thus also $|\sem{t}|$, is even} \\
                    & {} = \sem{\exp(-s,t)}
  \end{align*}
  The proofs for \ref{lem:symmetry-t-2} and \ref{lem:symmetry-t-3} are analogous, see
 App.\ \ref{sec:MissingProofs}.
\end{proof}
\makeproof*{lem:symmetry}{
  \symmetrysound*
	\begin{proof}
		\omitproof{Let $\AA \in \EIA$ again be arbitrary but fixed and let $\sem{\ldots}$ denote $\sem{\ldots}^\AA$.
		For \ref{lem:symmetry-t-1}, assume $\sem{\divisible_2(t)} = \true$, i.e., assume that $\sem{t}$ is even.
		Then it remains to show
		\[
			\sem{\exp(s,t) = \exp(-s,t)} = \true, \quad \text{i.e.,} \quad \sem{\exp(s,t)} = \sem{\exp(-s,t)}.
		\]
		We have:
		\begin{align*}
			\sem{\exp(s,t)} & {} = \sem{s}^{|\sem{t}|}                                                         \\
			                & {} = \sem{-s}^{|\sem{t}|} \tag{as $\sem{t}$, and thus also $|\sem{t}|$, is even} \\
			                & {} = \sem{\exp(-s,t)}
		\end{align*}
                }
                We already showed the soundness of \ref{lem:symmetry-t-1}
                in \Cref{sec:solving}.	
		For \ref{lem:symmetry-t-2}, assume $\sem{\divisible_2(t)} = \false$, i.e., assume that $\sem{t}$ is odd.
		Then it remains to show
		\[
			\sem{\exp(s,t) = -\exp(-s,t)} = \true, \quad \text{i.e.,} \quad \sem{\exp(s,t)} = \sem{-\exp(-s,t)}.
		\]
		We have:
		\begin{align*}
			\sem{\exp(s,t)} & {} = \sem{s}^{|\sem{t}|}                                                         \\
			                & {} = -\sem{-s}^{|\sem{t}|} \tag{as $\sem{t}$, and thus also $|\sem{t}|$, is odd} \\
			                & {} = -\sem{\exp(-s,t)}                                                           \\
			                & {} = \sem{-\exp(-s,t)}                                                           \\
		\end{align*}

		For \ref{lem:symmetry-t-3}, we have to show
		\[
			\sem{\exp(s,t) = \exp(s,-t)} = \true, \quad \text{i.e.,} \quad \sem{\exp(s,t)} = \sem{\exp(s,-t)}.
		\]
		We have:
		\[
			\sem{\exp(s,t)} = \sem{s}^{|\sem{t}|} = \sem{s}^{|-\sem{t}|} = \sem{s}^{|\sem{-t}|} = \sem{\exp(s,-t)}
		\]
	\end{proof}
}

\subsubsection{Monotonicity Lemmas}
\label{subsec:mon}

\emph{Monotonicity lemmas} are of the form
\begin{equation}
	\label{eq:monotonicity}
	\tag{\SC{mon}}
	s_2 \geq s_1 > 1 \land t_2 \geq t_1 > 0 \land (s_2 > s_1 \lor t_2 > t_1) \implies \exp(s_2,t_2) > \exp(s_1,t_1),
\end{equation}
i.e., they prohibit violations of monotonicity of $\exp$.

\begin{example}[Monotonicity Lemmas]
	For our leading example \eqref{eq:goal}, we obtain, e.g., the following lemmas:
	\begin{align}
		x > 1 \land \exp(y,y) > y^2 > 0 \implies {}  & \exp(x, \exp(y,y)) > \exp(x, y^2) \label{ex:mon_a}
		\\
		x > 1 \land -\exp(y,y) > y^2 > 0 \implies {} & \exp(x, -\exp(y,y)) > \exp(x, y^2) \label{ex:mon_b}
	\end{align}
\end{example}
So for each pair of two different relevant terms $\exp(s_1,t_1), \exp(s_2,t_2)$ where $\sem{s_2} \geq \sem{s_1} > 1$ and $\sem{t_2} \geq \sem{t_1} > 0$, the set $\LL$ contains \ref{eq:monotonicity}.
Here and in the sequel, unless mentioned otherwise, $\sem{\ldots}$ means $\sem{\ldots}^\AA$, where $\AA$ is the model from Line \ref{counterexample} of Alg.~\ref{alg}.
\begin{restatable}{lemma}{monotonicitysound}
	\label{lem:monotonicity}
	Let $s_1,s_2,t_1,t_2$ be terms of sort \normalfont{$\Int$}.
	Then \ref{eq:monotonicity} is \EIA-valid.
\end{restatable}
\begin{proof}
  Let $\AA \in \EIA$ be arbitrary but fixed and let $\sem{\ldots}$ denote $\sem{\ldots}^\AA$.
  Assume
  \[
    \sem{s_2 \geq s_1 > 1 \land t_2 \geq t_1 > 0 \land (s_2 > s_1 \lor t_2 > t_1)} = \true,
  \]
  i.e.,
  \[
    \sem{s_2} \geq \sem{s_1} > 1, \sem{t_2} \geq \sem{t_1} > 0, \text{ and } (\sem{s_2} > \sem{s_1} \text{ or } \sem{t_2} > \sem{t_1}).
  \]
  Then it remains to prove
  \[
    \sem{\exp(s_2,t_2) > \exp(s_1,t_1)} = \true, \quad \text{i.e.}, \quad \sem{\exp(s_2,t_2)} > \sem{\exp(s_1,t_1)}.
  \]
  First note that $\lambda x.\ x^{\sem{t_1}}$ is strictly monotonically increasing on $\NN_{>1}$ as $\sem{t_1} > 0$, and $\lambda y.\ \sem{s_1}^y$ is strictly monotonically increasing on $\NN_{>0}$ as $\sem{s_1} > 1$.
  Since we have $\sem{s_2} \in \NN_{>1}$ and $\sem{t_2} \in \NN_{>0}$, we get
  \[
    \sem{\exp(s_2,t_2)} = \sem{s_2}^{|\sem{t_2}|} > \sem{s_1}^{|\sem{t_1}|} = \sem{\exp(s_1,t_1)}
  \]
  due to monotonicity, as we have $\sem{s_2} \geq \sem{s_1}$ and $\sem{t_2} \geq \sem{t_1}$, where at least one of both inequations is strict.
\end{proof}

\subsubsection{Bounding Lemmas}

\emph{Bounding lemmas} provide bounds on relevant terms $\exp(s,t)$ where $\sem{s}$ and $\sem{t}$ are non-negative.
Together with symmetry lemmas, they also give rise to bounds for the cases where $s$ or $t$ are negative.

For each relevant term $\exp(s,t)$ where $\sem{s}$ and $\sem{t}$ are non-negative, the following lemmas are contained in $\LL$:
\begin{align}
	t = 0 \implies {}                             & \exp(s,t) = 1 \label{lem:constant-1}
	\tag{\protect{\ensuremath{\SC{bnd}_1}}}                                                          \\
	t = 1 \implies {}                             & \exp(s,t) = s \label{lem:constant-2}
	\tag{\protect{\ensuremath{\SC{bnd}_2}}}                                                          \\
	s = 0 \land t \neq 0 \iff {}                  & \exp(s,t) = 0 \label{lem:constant-t-1}
	\tag{\protect{\ensuremath{\SC{bnd}_3}}}                                                          \\
	s = 1 \implies {}                             & \exp(s,t) = 1 \label{lem:constant-t-2}
	\tag{\protect{\ensuremath{\SC{bnd}_4}}}                                                          \\
	s + t > 4 \land s > 1 \land t > 1 \implies {} & \exp(s,t) > s \cdot t + 1 \label{lem:constant-3}
	\tag{\protect{\ensuremath{\SC{bnd}_5}}}
\end{align}
The cases $t \in \{0,1\}$ are also addressed by our first rewrite rule (see \Cref{sec:preprocessings}).
However, this rewrite rule only applies if $t$ is an integer constant.
In contrast, the first two lemmas above apply if $t$ evaluates to $0$ or $1$ in the current model.
\begin{example}[Bounding Lemmas]
	\label{ex:bounding}
	For our leading example \eqref{eq:goal}, the following bounding lemmas would be considered, among others:
	\begin{align}
		\ref{lem:constant-1}:   &  & \exp(y,y) = 0 \implies {}            & \exp(x,\exp(y,y)) = 1 \notag             \\
		\ref{lem:constant-2}:   &  & \exp(y,y) = 1 \implies {}            & \exp(x,\exp(y,y)) = x \notag             \\
		\ref{lem:constant-t-1}: &  & x = 0 \land \exp(y,y) \neq 0 \iff {} & \exp(x,\exp(y,y)) = 0 \notag             \\
		\ref{lem:constant-t-2}: &  & x = 1 \implies {}                    & \exp(x,\exp(y,y)) = 1 \notag             \\
		\ref{lem:constant-3}:   &  & y > 2 \implies {}                    & \exp(y,y) > y^2 + 1 \label{ex:bound_a}
		\\
		\ref{lem:constant-3}:   &  & -y > 2 \implies {}                   & \exp(-y,-y) > y^2 + 1 \label{ex:bound_b}
	\end{align}
\end{example}
\begin{restatable}{lemma}{boundingsound}
	\label{lem:bounding}
	Let $s,t$ be terms of sort \normalfont{$\Int$}.
	Then \ref{lem:constant-1} -- \ref{lem:constant-3} are \EIA-valid.
\end{restatable}
\begin{proof}
  We only prove that \ref{lem:constant-3} is \EIA-valid and refer to
  App.\ \ref{sec:MissingProofs} for the (simpler) proofs 
  of \ref{lem:constant-1} -- \ref{lem:constant-t-2}.
  
  Let $\AA \in \EIA$ again be arbitrary but fixed and let $\sem{\ldots}$ denote $\sem{\ldots}^\AA$. 
  To show validity of \ref{lem:constant-3}, assume
  \[
    \sem{s + t > 4 \land s > 1 \land t > 1} = \true, \quad \text{i.e}, \quad \sem{s+t} > 4, \sem{s} > 1, \text{ and } \sem{t} > 1.
  \]
  Then it remains to show
  \[
    \sem{\exp(s,t) > s \cdot t + 1} = \true, \quad \text{i.e.}, \quad \sem{\exp(s,t)} > \sem{s} \cdot \sem{t} + 1.
  \]
  We use induction on $\sem{s} + \sem{t}$.
  In the base case, we have either $\sem{s} = 2$ and $\sem{t} = 3$, or $\sem{s} = 3$ and $\sem{t} = 2$.
  In the former case, we have:
  \[
    \sem{\exp(s,t)} = \sem{s}^{|\sem{t}|} = 2^{|3|} = 8 > 7 = 2 \cdot 3 + 1 = \sem{s} \cdot \sem{t} + 1
  \]
  In the latter case, we have:
  \[
    \sem{\exp(s,t)} = \sem{s}^{|\sem{t}|} = 3^{|2|} = 9 > 7 = 3 \cdot 2 + 1 = \sem{s} \cdot \sem{t} + 1
  \]
  For the induction step, assume $\sem{s} + \sem{t} > 5$.
  If $\sem{t} > 2$, then:
  \begin{align*}
    \sem{\exp(s,t)} & {} = \sem{s}^{|\sem{t}|}                                                                 \\
                    & {} = \sem{s}^{\sem{t}}                                                                   \\
                    & {} = \sem{s}^{\sem{t}-1} \cdot \sem{s}                                                   \\
                    & {} > (\sem{s} \cdot (\sem{t} - 1) + 1) \cdot \sem{s} \tag{by the induction hypothesis}   \\
                    & {} \geq (\sem{s} \cdot (\sem{t} - 1) + 1) \cdot 2 \tag{as $\sem{s} \geq 2$}              \\
                    & {} = (\sem{s} \cdot (\sem{t} - 1) + 1) + (\sem{s} \cdot (\sem{t} - 1) + 1)               \\
                    & {} \geq (\sem{s} \cdot (\sem{t} - 1) + 1) + (\sem{s} \cdot 2 + 1) \tag{as $\sem{t} > 2$} \\
                    & {} > (\sem{s} \cdot (\sem{t} - 1) + 1) + \sem{s}                                         \\
                    & {} = \sem{s} \cdot \sem{t} + 1
  \end{align*}
  If $\sem{t} = 2$, then $\sem{s} > 3$.
  Now we have
  \begin{align*}
    \sem{\exp(s,t)} & {} = \sem{s}^{|\sem{t}|}                                                            \\
                    & {} = \sem{s}^2                                                                      \\
                    & {} = (\sem{s-1} + 1)^2                                                              \\
                    & {} = \sem{s-1}^2 + 2 \cdot \sem{s-1} + 1                                            \\
                    & {} > \sem{s-1} \cdot 2 + 1 + 2 \cdot \sem{s-1} + 1\tag{by the induction hypothesis} \\
                    & {} \geq \sem{s-1} \cdot 2 + 8\tag{as $\sem{s} > 3$ and thus, $\sem{s-1} \geq 3$}    \\
                    & {} = \sem{s} \cdot 2 + 6                                                            \\
                    & {} = \sem{s} \cdot \sem{t} + 6                                                      \\
                    & {} > \sem{s} \cdot \sem{t} + 1
  \end{align*}
\end{proof}
\makeproof*{lem:bounding}{
  \boundingsound*
	\begin{proof}
	\omitproof{Let $\AA \in \EIA$ be arbitrary but fixed and let $\sem{\ldots}$ denote
          $\sem{\ldots}^\AA$.}
        We already showed EIA-validity of
      \ref{lem:constant-3}
                in \Cref{sec:solving}.	
		For \ref{lem:constant-1}, assume
		\[
			\sem{t=0} = \true, \quad \text{i.e.,} \quad \sem{t} = 0.
		\]
		Then it remains to show
		\[
			\sem{\exp(s,t) = 1} = \true, \quad \text{i.e.,} \quad \sem{\exp(s,t)} = 1.
		\]
		We have
		\[
			\sem{\exp(s,t)} = \sem{s}^{|\sem{t}|} = \sem{s}^{|0|} = \sem{s}^0 = 1.
		\]

		For \ref{lem:constant-2}, assume
		\[
			\sem{t=1} = \true, \quad \text{i.e.,} \quad \sem{t} = 1.
		\]
		Then it remains to show
		\[
			\sem{\exp(s,t) = s} = \true, \quad \text{i.e.,} \quad \sem{\exp(s,t)} = \sem{s}.
		\]
		We have
		\[
			\sem{\exp(s,t)} = \sem{s}^{|\sem{t}|} = \sem{s}^{|1|} = \sem{s}^1 = \sem{s}.
		\]

		For \ref{lem:constant-t-1}, we first show the implication from left to right.
		So assume
		\[
			\sem{s = 0 \land t \neq 0} = \true, \quad \text{i.e.,} \quad \sem{s} = 0 \text{ and } \sem{t} \neq 0.
		\]
		Then it remains to show
		\[
			\sem{\exp(s,t) = 0} = \true, \quad \text{i.e.,} \quad \sem{\exp(s,t)} = 0.
		\]
		We have
		\[
			\sem{\exp(s,t)} = \sem{s}^{|\sem{t}|} = 0^{|\sem{t}|} = 0
		\]
		where the last equality holds as we have $\sem{t} \neq 0$, and thus $|\sem{t}| > 0$.

		For the implication from right to left, assume
		\[
			\sem{\exp(s,t) = 0} = \true, \quad \text{i.e.,} \quad \sem{\exp(s,t)} = 0.
		\]
		Then it remains to show
		\[
			\sem{s = 0 \land t \neq 0} = \true, \quad \text{i.e.,} \quad \sem{s} = 0 \text{ and } \sem{t} \neq 0.
		\]
		We have
		\[
			\sem{\exp(s,t)} = \sem{s}^{|\sem{t}|} = 0.
		\]
		So as $\sem{s}^{|\sem{t}|} \neq 1$, we get $|\sem{t}| \neq 0$ and thus also $\sem{t} \neq 0$.
		As a product is only $0$ if one of its factors is $0$, we also get $\sem{s} = 0$.

		For \ref{lem:constant-t-2}, assume
		\[
			\sem{s=1} = \true, \quad \text{i.e.,} \quad \sem{s} = 1.
		\]
		Then it remains to show
		\[
			\sem{\exp(s,t) = 1} = \true, \quad \text{i.e.,} \quad \sem{\exp(s,t)} = 1.
		\]
		We have
		\[
			\sem{\exp(s,t)} = \sem{s}^{|\sem{t}|} = 1^{|\sem{t}|} = 1.
		\]
	\end{proof}
}

The reason why the bounding lemmas focus on lower bounds is that polynomials can only bound $\exp(s,t)$ from above for finitely many values of $s$ and $t$.
The bounding lemmas are defined in such a way that they provide lower bounds for $\exp(s,t)$ for almost all non-negative values of $s$ and $t$:
\ref{lem:constant-1}--\ref{lem:constant-t-2} provide tight bounds for all cases where at least one argument of $\exp$ is $0$ or $1$, and \ref{lem:constant-3} provides lower bounds for all other cases except for $\exp(2,2)$.
To cover \emph{all} cases, we could replace \ref{lem:constant-3}
by
\[
  s > 1 \land t \geq 0 \implies \exp(s,t) \geq s \cdot t.
\]
However, in contrast to the alternative lemma above, \ref{lem:constant-3} expresses that $\exp(s,t)$ is \emph{strictly} greater than $s \cdot t$.
The missing (lower and upper) bounds are provided by \emph{interpolation lemmas} (see \Cref{sec:interpolation}).

\subsubsection{Prime Lemmas}

\emph{Prime lemmas} are of the form
\begin{equation}
  \label{eq:prime}\tag{\SC{prime}}
  \divisible_d(\exp(s,t)) \iff \divisible_d(s)  \, \land \, t \neq 0 \qquad \text{where $d$ is prime}.
\end{equation}
So prime lemmas express that $\exp(s,t)$ and $s$ have the same prime factors for $t \neq 0$.
\begin{example}[Prime Lemmas]
  \label{ex:prime}
  Consider the formula
  \begin{equation}
    \label{eq:prime-ex}
    y \neq 0 \land \exp(2,x) = \exp(3,y),
  \end{equation}
  which was used by
  \textcite{conference}
  (where prime lemmas were not yet integrated) 
  to illustrate incompleteness of our approach.
  Unsatisfiability of this formula can easily be proven
    with the prime lemmas
  \begin{align*}
    \divisible_2(\exp(2,x)) \iff \divisible_2(2) \land x \neq 0& \;\; \equiv \;\;
    \divisible_2(\exp(2,x)) \iff x \neq 0& \text{and} \\
    \divisible_2(\exp(3,y)) \iff \divisible_2(3)  \land y \neq 0& \;\; \equiv \;\; \neg\divisible_2(\exp(3,y)),
  \end{align*}
  which rule out structures $\AA$ with $\AA \models x \neq 0$.
  The reason is that $x \neq 0$ implies $\divisible_2(\exp(2,x))$ by the first prime lemma, which, together with the second prime lemma, implies $\exp(2,x) \neq \exp(3,y)$.
  For the case $x = 0$, the following additional lemmas are needed:
  \begin{align*}
    x = 0 \implies \exp(2,x) & {} = 1 \tag{\ref{lem:constant-1}} \\
    y = 1 \implies \exp(3,y) & {} = 3 \tag{\ref{lem:constant-2}} \\
    y > 1 \implies \exp(3,y) & {} > 3 \cdot y + 1 \tag{\ref{lem:constant-3}} \\
    \exp(3,y) & {} = \exp(3,-y) \tag{\ref{lem:symmetry-t-3}}
  \end{align*}
  Then we have $\exp(2,x) = 1$ by the first bounding lemma and $\exp(3,y) > 1$ for all $y > 0$ by the second and third bounding lemma.
   Together with the symmetry lemma, the latter also implies $\exp(3,y) > 1$ for all $y < 0$.
   As \eqref{eq:prime-ex} implies $y \neq 0$, this suffices to prove unsatisfiability.
\end{example}
So for each relevant term $\exp(s,t)$ where $\sem{s}\geq 2$, $\sem{\exp(s,t)} \geq 2$,
and where
the sets
  of prime factors of $\sem{s}$ and $\sem{\exp(s,t)}$ differ, $\LL$ contains
\ref{eq:prime}, where $d$ is the smallest prime which divides $\sem{s}$ or
$\sem{\exp(s,t)}$, but not both of them.
To find $d$, we use wheel factorization \cite{wheel-factorization}.
\begin{lemma}
  \label{lem:prime}
  Let $s,t$ be terms of sort \normalfont{$\Int$} and let $d \in \NN$ be prime.
  Then \ref{eq:prime} is \EIA-valid.
\end{lemma}
\begin{proof}
  Again, let $\AA \in \EIA$ be arbitrary but fixed and let $\sem{\ldots}$ denote $\sem{\ldots}^\AA$.
  We need to prove
  \[
    \sem{\divisible_d(\exp(s,t))} = \sem{\divisible_d(s) \, \land \, t \neq 0},
  \]
  i.e.,
  \[
     \sem{\divisible_d(\exp(s,t))} \quad \text{ iff } \quad \sem{\divisible_d(s)} \text{ and
       } \sem{t} \neq 0.
  \]
  If $\sem{t} = 0$, then the claim holds since then $\sem{s}^{|\sem{t}|} = \sem{s}^{0} =
  1$ is not
  divisible by the prime number $d$.
  Otherwise, if $\sem{t} \neq 0$, then 
  we have to prove that $\sem{s}$ and $\sem{s}^{|\sem{t}|}$ have the same prime factors.
  If $P$ is the multiset of $\sem{s}$'s prime factors, i.e., $\prod P = \sem{s}$, then $\prod_{i=1}^{|\sem{t}|} \prod P = \sem{s}^{|\sem{t}|}$, so the claim follows.
\end{proof}

\subsubsection{Induction Lemmas}
\label{sec:induction}

\emph{Induction lemmas} are of the form
\begin{equation}
  \label{eq:induction}\tag{\SC{ind}}
  s_1 = s_2 \land t_2 - d = t_1 \geq 0 \implies \exp(s_2,t_2) = \exp(s_1,t_1) \cdot s_1^d \qquad \text{where } d \in \NN_{>0}.
\end{equation}
So intuitively, \ref{eq:induction} results from $\exp(s_2,t_2)$ by $d$ unrollings of the recursive definition
\[
  \exp(s,x) = \begin{cases}
    \exp(s,x-1) \cdot s & \text{if } x > 0\\
    1 & \text{if } x = 0
  \end{cases}
\]
as follows:
\[
  \exp(s_2,t_2) = \exp(s_2,t_2-1) \cdot s_2 = \exp(s_2,t_2-2) \cdot s_2^2 = \ldots = \exp(s_2,t_2-d) \cdot s_2^d = \exp(s_1,t_1) \cdot s_1^d
  \]

  \vspace*{-.1cm}
  
\begin{example}[Induction Lemmas]
  \label{ex:induction}
  We demonstrate the usefulness of induction lemmas by an example where SMT solving is
  used to verify a solution for a recurrence relation.
  Consider the following recurrence relation (where $n$ ranges over $\NN_{>0}$):
  \begin{equation}
    \label{eq:rec}
    f(n) = 2 \cdot f(n-1) + 2^n
  \end{equation}
  A possible solution (i.e., a function $f$ that satisfies \eqref{eq:rec}) is
  \begin{equation}
    \label{eq:sol}
    f(n) = (f(0) + n) \cdot 2^n.
  \end{equation}
  To verify that this function is indeed a solution, it suffices to replace $f$ with the definition from \eqref{eq:sol} and show that the resulting equation holds for all $n \in \NN_{>0}$ and all $f(0) \in \ZZ$, i.e., to prove
  \[
    \forall n \in \NN_{>0}, f_0 \in \ZZ.\ (f_0 + n) \cdot 2^n = 2 \cdot (f_0 + n - 1) \cdot 2^{n-1} + 2^n.
  \]
  To do so, it suffices to prove unsatisfiability of
  \begin{equation}
    \label{recExBenchmark}
    n \geq 1 \land (f_0 + n) \cdot \exp(2,n) \neq 2 \cdot (f_0 + n - 1) \cdot \exp(2,n-1) + \exp(2,n).
  \end{equation}
  To this end, our implementation deduces the following induction lemma:
  \[
    n \geq 1 \implies \exp(2,n) = 2 \cdot \exp(2,n-1)
  \]
  Then unsatisfiability can easily be proven without reasoning about $\exp$:
  Assume $n \geq 1$.
  Then by the induction lemma, we may substitute $\exp(2,n)$ with $2 \cdot \exp(2,n-1)$, resulting in
  \begin{align*}
    & (f_0 + n) \cdot 2 \cdot \exp(2,n-1) \neq 2 \cdot (f_0 + n - 1) \cdot \exp(2,n-1) + 2 \cdot \exp(2,n-1) \\
    {} \equiv {} & (f_0 + n) \cdot 2 \cdot \exp(2,n-1) \neq 2 \cdot (f_0 + n) \cdot \exp(2,n-1)
  \end{align*}
  which is trivially unsatisfiable.
  Hence, \eqref{eq:sol} is a valid solution for \eqref{eq:rec}.
\end{example}
So for each pair of relevant terms $\exp(s_1,t_1), \exp(s_2,t_2)$ where $\sem{s_1} = \sem{s_2}$ and $\sem{t_2} - d = \sem{t_1} \geq 0$ for some $d \in \NN_{>0}$, $\LL$ contains \ref{eq:induction}.
\begin{lemma}
  \label{lem:induction}
  Let $s_1,s_2,t_1,t_2$ be terms of sort \normalfont{$\Int$} and let $d \in \NN_{>0}$.
  Then \ref{eq:induction} is \EIA-valid.
\end{lemma}
\begin{proof}
  Let $\AA \in \EIA$ be an arbitrary but fixed model of
  \[
    s_1 = s_2 \land t_2 - d = t_1 \geq 0
  \]
  and let $\sem{\ldots}$ denote $\sem{\ldots}^\AA$.
  Then we have:
  \begin{align*}
    \sem{\exp(s_1,t_1) \cdot s_1^d} = {} & \sem{\exp(s_1,t_1)} \cdot \sem{s_1}^d \\
    {} = {} & \sem{s_1}^{|\sem{t_1}|} \cdot \sem{s_1}^d \\
    {} = {} & \sem{s_1}^{|\sem{t_1} + d|} \tag{as $\AA \models t_1 \geq 0$ and $d \in \NN_{>0}$} \\
    {} = {} & \sem{s_1}^{|\sem{t_2}|} \tag{as $\AA \models t_2 - d = t_1$} \\
    {} = {} & \sem{s_2}^{|\sem{t_2}|} \tag{as $\AA \models s_1 = s_2$} \\
    {} = {} & \sem{\exp(s_2,t_2)}
  \end{align*}
\end{proof}

\subsubsection{Interpolation Lemmas}
\label{sec:interpolation}
To provide bounds in cases where no bounding lemmas are violated, we use \emph{interpolation lemmas} that are constructed via \emph{bilinear interpolation}.
Here, we assume that the arguments of $\exp$ are positive, as negative arguments are handled by symmetry lemmas, and bounding lemmas yield tight bounds if at least one argument of $\exp$ is $0$.
The correctness of interpolation lemmas relies on the following observation.
As usual, $[w_1,w_2]$ and $(w_1,w_2)$ denote closed and open real intervals.
\begin{restatable}{lemma}{interpol}
	\label{lem:interpol}
	Let
          $c \geq 0$, $w_1,w_2 \in \RR_{> c}$,\footnote{The lemma also holds for $\RR_{\geq c}$ (and the proof is analogous).} $w_1 < w_2$, and let $f: \RR_{> c} \to \RR_{>0}$ be convex.
	Then
	\begin{align*}
		\forall x \in [w_1,w_2].                 & \ f(x) \leq f(w_1) + \frac{f(w_2)-f(w_1)}{w_2-w_1} \cdot (x - w_1)  & \text{and} \\
		\forall x \in \RR_{>c} \setminus (w_1,w_2). & \ f(x) \geq f(w_1) + \frac{f(w_2)-f(w_1)}{w_2-w_1} \cdot (x - w_1).
	\end{align*}
\end{restatable}
\noindent
Note that the right-hand side of the inequations above is the linear interpolant of $f$ between $w_1$ and $w_2$.
Intuitively, it corresponds to the secant of $f$ between the points $(w_1,f(w_1))$ and $(w_2,f(w_2))$, and thus the lemma follows from convexity of $f$.
\begin{proof}[Proof of \Cref{lem:interpol}]
  For the first inequation, recall that a function $f:\RR_{> c} \to \RR_{>0}$ is convex if for all $w_1,w_2 \in \RR_{> c}$ and all $v \in [0,1]$, we have
  \begin{equation}
    \label{eq:convex}
    f(v \cdot w_2 + (1-v) \cdot w_1) \leq v \cdot f(w_2) + (1-v) \cdot f(w_1).
  \end{equation}
  Let $x \in [w_1,w_2]$ and $v = \frac{x-w_1}{w_2 - w_1}$.
  Then we have $v \in [0,1]$ and
  \begin{align*}
    f(x) = {}  & f(v \cdot w_2 + (1-v) \cdot w_1)                                                \\
    {} \leq {} & v \cdot f(w_2) + (1-v) \cdot f(w_1) \tag{by convexity of $f$ and $v \in [0,1]$} \\
    {} = {}    & f(w_1) + \frac{f(w_2)-f(w_1)}{w_2-w_1} \cdot (x - w_1),
  \end{align*}
  as desired.

  For the second inequation, note that dually to \eqref{eq:convex}, convex functions $f: \RR_{> c} \to \RR_{>0}$ satisfy
  \[
    f(v \cdot w_2 + (1-v) \cdot w_1) \geq v \cdot f(w_2) + (1-v) \cdot f(w_1)
  \]
  for all $w_1,w_2 \in \RR_{> c}$ and $v \notin (0,1)$ such that $v \cdot w_2 + (1-v) \cdot w_1 \in \RR_{> c}$ \cite{reverse-jensen, reverse-jensen-2}.
  Let $x \in \RR_{>c} \setminus (w_1,w_2)$ and $v = \frac{x-w_1}{w_2 - w_1}$.
  Then we have $v \notin (0,1)$ and
  \begin{align*}
    f(x) = {}  & f(v \cdot w_2 + (1-v) \cdot w_1)                                                   \\
    {} \geq {} & v \cdot f(w_2) + (1-v) \cdot f(w_1) \tag{by convexity of $f$ and $v \notin (0,1)$} \\
    {} = {}    & f(w_1) + \frac{f(w_2)-f(w_1)}{w_2-w_1} \cdot (x - w_1),
  \end{align*}
  as desired.
\end{proof}

Let $\exp(s,t)$ be relevant, $\sem{s} = c > 0$, $\sem{t} = d > 0$, and ${\sem{\exp}}(c,d) \neq c^d$, i.e., we want to prohibit the current interpretation of $\exp(s,t)$.

\paragraph{Interpolation Lemmas for Upper Bounds}
First assume ${\sem{\exp}}(c,d) > c^d$, i.e., to rule out this counterexample, we need a lemma that provides a suitable upper bound for $\exp(c,d)$.
Let $c',d' \in \NN_{>0}$ and:
\begin{align*}
	c^- & {} \Def \min(c,c') & c^+     & {} \Def \max(c,c')          & d^-     & {} \Def \min(d,d')          & d^+ & {} \Def \max(d,d') \\
	    &                    & [c^\pm] & {} \Def [c^- \twoldots c^+] & [d^\pm] & {} \Def [d^- \twoldots d^+]
\end{align*}
Here, $[a \twoldots b]$ denotes a closed integer interval.
Then we first use $d^-,d^+$ for linear interpolation w.r.t.\ the $2^{nd}$ argument of $\lambda x,y.\ x^y$.
To this end, let
\[
	\ip{2}{d}(x,y) \Def x^{d^-} + \efrac{x^{d^+} - x^{d^-}}{d^+ - d^-} \cdot (y - d^-),
\]
where we define $\efrac{a}{b} \Def \frac{a}{b}$ if $b \neq 0$ and $\efrac{a}{0} \Def 0$.
So if $d^- < d^+$, then $\ip{2}{d}(x,y)$ corresponds to the linear interpolant of $x^y$ w.r.t.\ $y$ between $d^-$ and $d^+$.
Then $\ip{2}{d}(x,y)$ is a suitable upper bound, as
\begin{equation}
	\label{eq:secant}
	\forall x \in \NN_{>0}, y \in [d^\pm].\ x^y \leq \ip{2}{d}(x,y)
\end{equation}
follows from \Cref{lem:interpol}:
If $d^- = d^+$, then \eqref{eq:secant} is trivial.
Otherwise, let $x \in \NN_{>0}$ be arbitrary but fixed.
Then \eqref{eq:secant} follows by applying \Cref{lem:interpol} to $f(y) \Def x^y$, which is clearly convex on $\RR_{>0}$.

Hence, we could derive the following \EIA-valid lemma:\footnote{Strictly speaking, this
lemma is not a $\Sigma_\Int^\exp$-term if $d^+ > d^-$, as the right-hand side makes use of
(rational) division in this case.
	However, an equivalent $\Sigma_\Int^\exp$-term can clearly be obtained by multiplying with the divisor.}
\begin{equation}
	s > 0 \land t \in [d^\pm] \implies \exp(s,t) \leq \ip{2}{d}(s,t) \tag{\protect{\ensuremath{\SC{ip}_1}}}\label{eq:il1}
\end{equation}

\begin{example}[Linear Interpolation w.r.t.\ $y$]
	\label{ex:linear-interpolation}
	Let $\sem{\exp(s,t)} = {\sem{\exp}}(3,9) > 3^9$, i.e., we have $c = 3$ and $d = 9$.
	Moreover, assume $c' = d' = 1$, i.e., we get $c^- = 1$, $c^+ = 3$, $d^- = 1$, and $d^+ = 9$.
	Then
	\[
		\ip{2}{d}(x,y) = \ipc{2}{1}{9}(x,y) = x^{1} + \frac{x^{9} - x^{1}}{9 - 1} \cdot (y - 1) = x + \frac{x^9-x}{8} \cdot (y-1).
	\]
	Hence, \ref{eq:il1} corresponds to
	\begin{align*}
		s > 0 \land t \in [1,9] \implies \exp(s,t) \leq s + \frac{s^9-s}{8} \cdot (t-1).
	\end{align*}
	This lemma would be violated by our counterexample, as we have
	\begin{align*}
		\sem{s + \frac{s^9-s}{8} \cdot (t-1)} = 3 + \frac{3^9-3}{8} \cdot 8 = 3^9 < {\sem{\exp}}(3,9) = \sem{\exp(s,t)}.
	\end{align*}
\end{example}
However, the degree of $\ip{2}{d}(s,t)$ depends on $d^+$, which in turn depends on the model that was found by the underlying SMT solver.
Thus, the degree of $\ip{2}{d}(s,t)$ can get very large, which is challenging for the underlying solver.

So we next use $c^-,c^+$ for linear interpolation w.r.t.\ the $1^{st}$ argument of $\lambda x,y.\ x^y$, resulting in
\[
	\ip{1}{c}(x,y) \Def (c^-)^y + \efrac{(c^+)^y - (c^-)^y}{c^+ - c^-} \cdot (x - c^-).
\]
Then due to \Cref{lem:interpol}, $\ip{1}{c}(x,y)$ is also an upper bound on the exponentiation function, i.e., we have
\begin{equation}
	\label{eq:secant1}
	\forall y \in \NN_{>0}, x \in [c^\pm].\ x^y \leq \ip{1}{c}(x,y).
\end{equation}
To see this, note that \eqref{eq:secant1} is trivial if $c^- = c^+$.
Otherwise, let $y \in \NN_{>0}$ be arbitrary but fixed.
Then \eqref{eq:secant1} follows by applying \Cref{lem:interpol} to $f(x) \Def x^y$, which is clearly convex on $\RR_{>0}$.

Note that we have $\efrac{y - d^-}{d^+ - d^-} \in [0,1]$ for all $y \in [d^\pm]$, and thus
\begin{align*}
	\ip{2}{d}(x,y) = x^{d^-} \cdot \left(1-\efrac{y - d^-}{d^+ - d^-}\right) + x^{d^+} \cdot \efrac{y - d^-}{d^+ - d^-}
\end{align*}
is monotonically increasing in both $x^{d^-}$ and $x^{d^+}$.
Hence, in the definition of $\ip{2}{d}$, we can approximate $x^{d^-}$ and $x^{d^+}$ with their upper bounds $\ip{1}{c}(x,d^-)$ and $\ip{1}{c}(x,d^+)$ that can be derived from \eqref{eq:secant1}.
Then \eqref{eq:secant} yields
\begin{equation}
	\label{eq:secant2}
	\forall x \in [c^\pm], y \in [d^\pm].\ x^y \leq \ipf{c}{d}(x,y)
\end{equation}
where
\[
	\ipf{c}{d}(x,y) \Def \ip{1}{c}(x,d^-) + \efrac{\ip{1}{c}(x,d^+) - \ip{1}{c}(x,d^-)}{d^+ - d^-} \cdot (y - d^-).
\]
So the set $\LL$ contains the lemma
\begin{equation}
	s \in [c^\pm] \land t \in [d^\pm] \implies \exp(s,t) \leq \ipf{c}{d}(s,t), \tag{\protect{\ensuremath{\SC{ip}_2}}}\label{eq:il2}
\end{equation}
which is valid due to \eqref{eq:secant2}, and rules out any counterexample with
${\sem{\exp}}(c,d) > c^d$, as $\ipf{c}{d}(c,d) = c^d$.

\begin{lemma}
	\label{lem:SoundnessInterpolationUpper}
	Let $c^+ \geq c^- > 0$ and $d^+ \geq d^- > 0$.
	Then \ref{eq:il2} is \EIA-valid.
\end{lemma}

\begin{example}[Bilinear Interpolation, \Cref{ex:linear-interpolation} continued]
	\label{ex:ip2}
	In our example, we have:
	\begin{align*}
		\ip{1}{c}(x,y)   & {} = \ipc{1}{1}{3}(x,y) = 1^y + \frac{3^y - 1^y}{3 - 1} \cdot (x - 1) = 1 + \frac{3^y-1}{2} \cdot (x-1) \\
		\ip{1}{c}(s,d^-) & {} = \ipc{1}{1}{3}(s,1) = 1 + \frac{3-1}{2} \cdot (s-1) = s                                             \\
		\ip{1}{c}(s,d^+) & {} = \ipc{1}{1}{3}(s,9) = 1 + \frac{3^9-1}{2} \cdot (s-1) = 1 + 9841 \cdot (s-1)
	\end{align*}
	Hence, we obtain the lemma
	\[
		s \in [1,3] \land t \in [1,9] \implies \exp(s,t) \leq s + \frac{1 + 9841 \cdot (s - 1) -s}{8} \cdot (t-1).
	\]
	This lemma is violated by our counterexample, as we have
	\begin{align*}
		\sem{s + \frac{1 + 9841 \cdot (s - 1) -s}{8} \cdot (t-1)} = 3^9 < {\sem{\exp}}(3,9) = \sem{\exp(s,t)}.
	\end{align*}
\end{example}
\ref{eq:il2} relates $\exp(s,t)$ with the \emph{bilinear} function $\ipf{c}{d}(s,t)$, i.e., this function is linear w.r.t.\ both $s$ and $t$, but it multiplies $s$ and $t$.
Thus, if $s$ is an integer constant and $t$ is linear, then the resulting lemma is linear, too.

To compute interpolation lemmas, a second point $(c',d')$ is needed.
In our implementation, we store all points $(c,d)$ where interpolation has previously been applied and use the one which is closest to the current one.
The same heuristic is used by \textcite{cimatti18a} to compute \emph{secant lemmas}.
For the $1^{st}$ interpolation step, we choose $(c',d') = (c,d)$.
In this case, \ref{eq:il2} simplifies to $s = c \land t = d \implies \exp(s,t) \leq c^d$.

\paragraph{Interpolation Lemmas for Lower Bounds}

While bounding lemmas already yield lower bounds, the bounds provided by \ref{lem:constant-3} are not exact, in general.
Hence, if ${\sem{\exp}}(c,d) < c^d$, then we also use bilinear interpolation to obtain a precise lower bound for $\exp(c,d)$.
Dually to \eqref{eq:secant} and \eqref{eq:secant1}, \Cref{lem:interpol} implies:
\begin{align}
  \label{eq:lb1}
  \forall x,y \in \NN_{>0}.\ x^y \geq \ipc{2}{d}{d+1}(x,y)\\
  \label{eq:lb2a}
  \forall x,y \in \NN_{>0}.\ x^y \geq \ipc{1}{c}{c+1}(x,y)
\end{align}
To see this for \eqref{eq:lb1}, let $x \in \NN_{>0}$ be arbitrary but fixed.
Then \eqref{eq:lb1} follows by applying \Cref{lem:interpol} to $f(y) \Def x^y$,
which is clearly convex on $\RR_{>0}$. While \Cref{lem:interpol} implies the claim for $y
\in \RR_{>0} \setminus (d, d+1)$, note that
  $\NN_{>0} \subseteq \RR_{>0} \setminus (d, d+1)$, since $(d, d+1)$ does not contain any
integers. 
The proof of \eqref{eq:lb2a} works analogously by considering $f(x) \Def x^y$ for
fixed $y \in \NN_{>0}$.

Additionally, we also obtain
\begin{equation}
	\label{eq:lb2}
	\forall x,y \in \NN_{>0}.\ x^{y+1} - x^y \geq \ipc{1}{c}{c+1}(x,y+1)-\ipc{1}{c}{c+1}(x,y)
\end{equation}
from \Cref{lem:interpol}.
The reason is that for $f(x) \Def x^{y+1} - x^y$, the right-hand side of \eqref{eq:lb2} is equal to the linear interpolant of $f$ between $c$ and $c+1$.
Moreover, $f$ is convex, as $f(x) = x^y \cdot (x-1)$ where for any fixed $y \in \NN_{>0}$, both $x^y$ and $x-1$ are non-negative, monotonically increasing, and convex on $\RR_{\geq 1}$.
Thus, \eqref{eq:lb2} follows by applying \Cref{lem:interpol} to $f(x) \Def x^{y+1} - x^y$,
since $\NN_{>0} \subseteq \RR_{\geq 1} \setminus (c, c+1)$.

If $y \geq d$, then $\ipc{2}{d}{d+1}(x,y) = x^{d} + (x^{d+1} - x^{d}) \cdot (y - d)$ is monotonically increasing in the first occurrence of $x^{d}$, and in $x^{d+1} - x^{d}$.
Thus, by approximating $x^{d}$ and $x^{d+1} - x^{d}$ with their lower bounds from \eqref{eq:lb2a} and \eqref{eq:lb2}, \eqref{eq:lb1} yields
\begin{align}
	\nonumber \forall x \in \NN_{>0}, y \geq d.\ x^y & \geq \ipc{1}{c}{c+1}(x,d) + (\ipc{1}{c}{c+1}(x,d+1) - \ipc{1}{c}{c+1}(x,d)) \cdot (y-d) \\
	\label{eq:lb3}
	                                              & = \ipfc{c}{c+1}{d}{d+1}(x,y).
\end{align}
So dually to \ref{eq:il2}, the set $\LL$ contains the lemma
\begin{equation}
	\tag{\protect{\ensuremath{\SC{ip}_3}}}\label{eq:il5}
	s \geq 1 \land t \geq d \implies \exp(s,t) \geq \ipfc{c}{c+1}{d}{d+1}(s,t)
\end{equation}
which is valid due to \eqref{eq:lb3} and rules out any counterexample with ${\sem{\exp}}(c,d) < c^d$, as $\ipfc{c}{c+1}{d}{d+1}(c,d) = c^d$.

\begin{lemma}
	\label{lem:SoundnessInterpolationLower}
	Let $c,d \in \NN_{>0}$.
	Then \ref{eq:il5} is \EIA-valid.
\end{lemma}

\begin{example}[Interpolation, Lower Bounds]
	\label{ex:bilinear-interpolation-lower}
	Let $\sem{\exp(s,t)} = {\sem{\exp}}(3,9) < 3^9$, i.e., we have $c = 3$, and $d = 9$.
	Then
	\begin{align*}
		\ipc{1}{3}{4}(x,9)      & {} = 3^9 + (4^9-3^9) \cdot (x - 3)
		{} = 19683 + 242461 \cdot (x-3)                                                                            \\
		\ipc{1}{3}{4}(x,10)     & {} = 3^{10} + (4^{10}-3^{10}) \cdot (x - 3)
		{} = 59049 + 989527 \cdot (x - 3)                                                                          \\
		\ipfc{3}{4}{9}{10}(x,y) & {} = \ipc{1}{3}{4}(x,9) + (\ipc{1}{3}{4}(x,10) - \ipc{1}{3}{4}(x,9)) \cdot (y-9)
	\end{align*}
	and thus we obtain the lemma
	\[
		s \geq 1 \land t \geq 9 \implies \exp(s,t) \geq 747066 \cdot s \cdot t - 6481133 \cdot s- 2201832 \cdot t + 19108788.
	\]
	It is violated by our counterexample, as we have
	\[
		\sem{747066 \cdot s \cdot t - 6481133 \cdot s - 2201832 \cdot t + 19108788} = 3^9 > {\sem{\exp}}(3,9).
	\]
\end{example}

\subsection{Lazy Lemma Generation}
\label{subsec:lazy}
In practice, it is not necessary to compute the entire set of lemmas $\LL$.
Instead, we can stop as soon as $\LL$ contains a single lemma which is violated by the current counterexample.
However, such a strategy would result in a quite fragile implementation, as its behavior would heavily depend on the order in which lemmas are computed, which in turn depends on low-level details like the order of iteration over sets, etc.
So instead, we use the following precedence on our six kinds of lemmas:
\[
	\text{symmetry} \succ \text{monotonicity} \succ \text{bounding} \succ \text{prime}
        \approx \text{induction} \approx \text{interpolation}
\]
Then we compute all lemmas of the same kind, starting with symmetry lemmas, and we only proceed with kinds of lower precedence if none of the lemmas computed so far is violated by the current counterexample.
The motivation for the order above is as follows:
First, we prefer symmetry, monotonicity, and bounding lemmas, as there are only finitely
many of them for a finite set of relevant terms. In contrast, in principle there could be
infinitely many  prime, induction, and interpolation lemmas, so they all get the same, lowest precedence.
The reason is that the natural numbers $c, d, c^-, c^+, d^-, d^+$ that occur in these lemmas depend on the current \NIA-model $\AA$.
Symmetry lemmas obtain the highest precedence, as other kinds of lemmas depend on them for restricting $\exp(s,t)$ in the case that $s$ or $t$ is negative.
Moreover, we prefer monotonicity lemmas over bounding lemmas, as monotonicity lemmas are
linear (if the arguments of $\exp$ are linear), whereas \ref{lem:constant-3} may be
non-linear.

\begin{algorithm}[t]
	\caption{CEGAR for $\EIA$ with Lazy Lemma Generation}\label{alg2}
	\KwIn{a $\Sigma_\Int^\exp$-formula $\phi$}
	\tcp{Preprocessing}
	\Do{$\phi \neq \phi'$}{
		$\phi' \gets \phi$\;
		$\phi \gets \SC{FoldConstants}(\phi)$\;
		$\phi \gets \SC{Rewrite}(\phi)$\;
	}
	\tcp{Refinement Loop}
	\While{there is a \NIA-model $\AA$ of $\phi$}{
		\If{$\AA$ is a counterexample}{
			$\LL \gets \emptyset$\;
			\For{$\mathrm{k} \in \{ \text{Symmetry}, \text{Monotonicity}, \text{Bounding}\}$ \label{alg2:loop1}}{
				$\LL \assign \{\psi \in \SC{ComputeLemmas}(\phi, \mathrm{k}) \mid \AA \not\models \psi\}$\;
                                \lIf{$\LL \neq \emptyset$}{{\bf break}}
			}
                        \If{$\LL = \emptyset$ \label{alg2:check}}{
                          \For{$\mathrm{k} \in \{ \text{Prime}, \text{Induction}, \text{Interpolation}\}$ \label{alg2:loop2}}{
                            $\LL \assign \LL \cup \{\psi \in \SC{ComputeLemmas}(\phi, \mathrm{k}) \mid \AA \not\models \psi\}$\;
                          }
                        }
			$\phi \gets \phi \land \bigwedge \LL$
		}
		\lElse{
			\Return{$\sat$}
		}
	}
	\Return{$\unsat$}
\end{algorithm}

The improved version of Alg.~\ref{alg} is shown in Alg.~\ref{alg2}.
Note that the first inner loop in Line~\ref{alg2:loop1} breaks as soon as a lemma that is violated by the current model $\AA$ has been found, i.e., as soon as $\LL$ is non-empty.
Then due to the check in Line~\ref{alg2:check}, the second inner loop in Line~\ref{alg2:loop2} is only executed if the first inner loop failed to eliminate the current counterexample.
In contrast to the first inner loop, the second one in Line~\ref{alg2:loop2} does not break, as prime, induction, and interpolation lemmas have the same precedence.

\begin{example}[Leading Example Finished]
	\label{Leading Example -- Total Semantics}
	We now finish our leading example which, after preprocessing, looks as follows (see \Cref{ex:preprocessing}):
	\begin{align}
		x^2 > 4 \land y^2 > 4 \land \exp(x,y^2) = \exp(x,\exp(y,y)) \tag{\ref{eq:goal}}
	\end{align}
	Then our implementation generates $12$ symmetry lemmas, $4$ monotonicity lemmas, and $17$ bounding lemmas before proving unsatisfiability, including
	\[
		\eqref{ex:sym_a}, \eqref{ex:sym_b}, \eqref{ex:sym_c}, \eqref{ex:sym_d}, \eqref{ex:mon_a}, \eqref{ex:mon_b}, \eqref{ex:bound_a}, \text{ and } \eqref{ex:bound_b}.
	\]
	These lemmas suffice to prove unsatisfiability for the case $x > 2$ (the cases $x \in [-2 \twoldots 2]$ or $y \in [-2 \twoldots 2]$ are trivial).
	For example, if $y < -2$ and $\divisible_2(-y)$, we get
	\begin{align*}
		% y^2 > 4 & {} \curvearrowright -y < 2
		y < -2 & {}\overset{\eqref{ex:bound_b}}{\curvearrowright} \exp(-y,-y) > y^2 + 1 \overset{\eqref{ex:sym_a}}{\curvearrowright} \exp(y,-y) > y^2 + 1                                                         \\
		       & \overset{\eqref{ex:sym_d}}{\curvearrowright} \exp(y,y) > y^2 + 1 \overset{\eqref{ex:mon_a}}\curvearrowright \exp(x, \exp(y,y)) > \exp(x, y^2) \overset{\eqref{eq:goal}}{\curvearrowright} \false
	\end{align*}
	and for the cases $y>2$ and $y < -2 \land \neg\divisible_2(-y)$, unsatisfiability can be shown similarly.
	For the overall proof, $5$ more symmetry lemmas, $2$ more monotonicity lemmas, and $4$ more bounding lemmas are used.
	The remaining lemmas that are deduced by our implementation are not used in the final proof of unsatisfiability.
\end{example}

While our leading example can be solved without interpolation lemmas, in general, interpolation lemmas are a crucial ingredient of our approach.
\begin{example}
	Consider the formula
	\[
		1 < x < y \land 0 < z \land \exp(x,z) < \exp(y,z).
	\]
	Our implementation first rules out $7$ counterexamples using $7$ bounding lemmas, $15$ prime lemmas, and $9$ interpolation lemmas, before finding the model $\sem{x} = 5$, $\sem{y} = 7$, and $\sem{z} = 3$.
	Without interpolation lemmas, our implementation keeps computing prime lemmas
        indefinitely, and without prime and interpolation lemmas, our implementation
        returns \unknown.\footnote{Note that without interpolation lemmas, Alg.~\ref{alg2} may fail to eliminate a counterexample, as \Cref{ProgressTheorem} does not hold without interpolation lemmas.
        In such cases, our implementation returns \unknown.}
\end{example}
Our main soundness theorem follows from soundness of our preprocessings (\Cref{lem:preproc}) and the fact that all of our lemmas are \EIA-valid (Lemmas \ref{lem:symmetry}, \ref{lem:monotonicity}, \ref{lem:bounding}, \ref{lem:prime}, \ref{lem:induction}, \ref{lem:SoundnessInterpolationUpper}, and \ref{lem:SoundnessInterpolationLower}).
\begin{theorem}[Soundness of Alg.~\ref{alg2}]
	\label{thm:sound}
	If Alg.~\ref{alg2} returns $\sat$, then $\phi$ is satisfiable in $\EIA$.
	If Alg.~\ref{alg2} returns $\unsat$, then $\phi$ is unsatisfiable in $\EIA$.
\end{theorem}
Another important property of Alg.~\ref{alg2} is that it can eliminate \emph{any} counterexample, and hence it makes progress in every iteration.
\begin{restatable}[Progress Theorem]{theorem}{progress}
	\label{ProgressTheorem}
	If $\AA$ is a counterexample and $\LL$ is computed as in Alg.~\ref{alg2}, then
	\[
		\AA \not\models \bigwedge \LL.
	\]
\end{restatable}
\begin{proof}
  Let $\AA \in \EIA$ be arbitrary but fixed and let $\sem{\ldots}$ denote $\sem{\ldots}^\AA$.
  Since $\AA$ is a counterexample, there is a subexpression $\exp(s,t)$ of $\phi$ such that $\sem{\exp(s,t)} \neq \sem{\exp(s,t)}^\EIA$.
  We now show that $\AA$ violates at least one of our lemmas.

  First assume $\sem{s}, \sem{t} \in \NN$.
  Then
  \begin{itemize}
  \item $t = 0$ implies $\sem{\ref{lem:constant-1}} = \false$ and
  \item $s = 0 \land t \neq 0$ implies $\sem{\ref{lem:constant-t-1}} = \false$.
  \end{itemize}
  Hence, assume $\sem{s}, \sem{t} \in \NN_{>0}$.
  Then
  \begin{itemize}
  \item $\sem{\exp(s,t)} > \sem{\exp(s,t)}^\EIA$ implies $\sem{\ref{eq:il2}} = \false$ (as remarked after \ref{eq:il2}), and
  \item $\sem{\exp(s,t)} < \sem{\exp(s,t)}^\EIA$ implies $\sem{\ref{eq:il5}} = \false$ (as remarked after \ref{eq:il5}).
  \end{itemize}
  Now assume $\sem{s} < 0$ and $\sem{t} \geq 0$.
  Since $\exp(-s,t)$ is relevant, the set $\LL$ contains bounding and interpolation lemmas for $\exp(-s,t)$.
  Hence, if $\sem{\exp(-s,t)} \neq \sem{\exp(-s,t)}^\EIA$, then one of these lemmas is violated by $\AA$, as argued above.
  Thus, assume $\sem{\exp(-s,t)} = \sem{\exp(-s,t)}^\EIA$.
  Then
  \begin{itemize}
  \item if $\sem{t}$ is even, then $\sem{\ref{lem:symmetry-t-1}} = \false$, and
  \item if $\sem{t}$ is odd, then $\sem{\ref{lem:symmetry-t-2}} = \false$.
  \end{itemize}
  Next assume $\sem{s} \geq 0$ and $\sem{t} < 0$.
  Since $\exp(s,-t)$ is relevant, the set $\LL$ contains bounding and interpolation lemmas for $\exp(s,-t)$.
  Hence, if $\sem{\exp(s,-t)} \neq \sem{\exp(s,-t)}^\EIA$, then one of these lemmas is violated by $\AA$, as argued above.
  Thus, assume $\sem{\exp(s,-t)} = \sem{\exp(s,-t)}^\EIA$.
  Then $\sem{\ref{lem:symmetry-t-3}} = \false$.

  Finally, assume $\sem{s_1} < 0$ and $\sem{t_1} < 0$.
  Since $\exp(-s,-t)$ is relevant, the set $\LL$ contains bounding and interpolation lemmas for $\exp(-s,-t)$.
  Hence, if $\sem{\exp(-s,-t)} \neq \sem{\exp(-s,-t)}^\EIA$, then one of these lemmas is violated by $\AA$, as argued above.
  Thus, assume $\sem{\exp(-s,-t)} = \sem{\exp(-s,-t)}^\EIA$.

  We only consider the case that $\sem{t}$ is even (the case that $\sem{t}$ is odd works analogously).
  Assume $\sem{\ref{lem:symmetry-t-1}} = \true$ and $\sem{\ref{lem:symmetry-t-3}} = \true$.
  Then we get
  \begin{align*}
    \sem{\exp(s,t)}^\EIA & {} = \sem{\exp(-s,t)}^\EIA \tag{as $t$ is even}                      \\
                         & {} = \sem{\exp(-s,-t)}^\EIA \tag{by definition of $\sem{\exp}^\EIA$} \\
                         & {} = \sem{\exp(-s,-t)}                                               \\
                         & {} = \sem{\exp(s,-t)} \tag{\ref{lem:symmetry-t-1}}                   \\
                         & {} = \sem{\exp(s,t)} \tag{\ref{lem:symmetry-t-3}},
  \end{align*}
  which contradicts $\sem{\exp(s,t)} \neq \sem{\exp(s,t)}^\EIA$.
  Thus, we have $\sem{\ref{lem:symmetry-t-1}} = \false$ or $\sem{\ref{lem:symmetry-t-3}}
  = \false$.
\end{proof}

Despite \Cref{thm:sound,ProgressTheorem}, \EIA is of course undecidable, and hence Alg.~\ref{alg2} is incomplete.
For example, it does not terminate for the input formula
\begin{equation}
	\label{incompleteFormula}
	x \geq y \geq 0 \land \exp(2,x) \neq \exp(2,x-y) \cdot \exp(2,y).
\end{equation}
Here, to prove unsatisfiability,\footnote{Note that induction lemmas are not sufficient for proving unsatisfiability of \eqref{incompleteFormula}, as the difference between the exponents in induction lemmas is always a constant, whereas it is non-constant in \eqref{incompleteFormula}.} one would need a rewrite rule like
\[
  \exp(x,y) \cdot \exp(x,z) \to \exp(x,y+z),
\]
but such a rule would be unsound in our setting, as the right-hand side would need to be $\exp(x,|y|+|z|)$ instead.
However, to avoid introducing non-linear terms like $|y|$ and $|z|$, we did not include
such a rewrite rule.
Thus, Alg.~\ref{alg2} would refine the formula \eqref{incompleteFormula} infinitely often.

Note that monotonicity, prime, and induction lemmas are important, even though they are not required to prove \Cref{ProgressTheorem}.
The reason is that \emph{all} (usually infinitely many) counterexamples must be eliminated to prove $\unsat$.
For instance, reconsider \Cref{Leading Example -- Total Semantics}, where the monotonicity lemma \eqref{ex:mon_a} eliminates infinitely many counterexamples with $\sem{\exp(x, \exp(y,y))} \leq \sem{\exp(x, y^2)}$.
In contrast, \Cref{ProgressTheorem} only guarantees that every single counterexample can be eliminated.
Consequently, our implementation does not terminate on our leading example if monotonicity lemmas are disabled.
Similarly, our implementation fails to solve \Cref{ex:prime} and \Cref{ex:induction} without prime and induction lemmas, respectively.

\section{Phasing}
\label{sec:phasing}

\begin{algorithm}[t]
	\caption{CEGAR for $\EIA$ with Phasing}\label{alg3}
	\KwIn{a $\Sigma_\Int^\exp$-formula $\phi$}
	\tcp{Preprocessing}
	\Do{$\phi \neq \phi'$}{
		$\phi' \gets \phi$\;
		$\phi \gets \SC{FoldConstants}(\phi)$\;
		$\phi \gets \SC{Rewrite}(\phi)$\;
	}
        $\mathit{phase} \gets \sat$\;
        $b \gets 1$\;
	\tcp{Refinement Loop}
        \While{$\top$}{
          \leIf{$\mathit{phase} = \sat$\label{alg3:bound}}{$\xi \gets \phi \land \{-2^b \leq t \leq 2^b \mid \exp(s,t) \text{ is relevant}\}$}{$\xi \gets \phi$}
          \If{there is a \NIA-model $\AA$ of $\xi$}{
            \If{$\AA$ is a counterexample}{
              \If{$\mathit{phase} = \unsat$\label{alg3:discard}} {
                $\mathit{phase} \gets \sat$\;\label{alg3:switch-sat}
                $b\inc$\label{alg3:inc}
              }
              \Else{
                \For{$\mathrm{k} \in \{ \text{Symmetry}, \text{Monotonicity},
                  \text{Bounding}, \text{Prime} \}$}{
                  $\LL \assign \{\psi \in \SC{ComputeLemmas}(\phi, \mathrm{k}) \mid \AA \not\models \psi\}$\;
                  \lIf{$\LL \neq \emptyset$}{{\bf break}}
                }
                \If{$\LL = \emptyset$}{
                  \For{$\mathrm{k} \in \{\text{Induction}, \text{Interpolation}\}$}{
                    $\LL \assign \LL \cup \{\psi \in \SC{ComputeLemmas}(\phi, \mathrm{k}) \mid \AA \not\models \psi\}$\;
                  }
                }
                $\phi \gets \phi \land \bigwedge \LL$
              }
            }
            \lElse{
              \Return{$\sat$}
            }
          } \lElseIf{$\mathit{phase} = \unsat$}{
            \Return{$\unsat$}
          } \lElse{
            $\mathit{phase} \gets \unsat$\label{alg3:switch-unsat}
          }
        }
\end{algorithm}

When investigating the behavior of our implementation for satisfiable examples where it
failed to find a solution, we frequently observed the following problem:
Even though a ``small'' solution exists where the exponents (i.e., the second arguments of
the occurring $\exp$-terms) have small absolute values, the underlying SMT solver returned
candidate models with exponents whose absolute values were quite large.
Then interpolation lemmas with huge coefficients were generated, as these coefficients grow exponentially w.r.t.\ the second argument of $\exp$.
Afterwards, the underlying solver took too long when searching for a new candidate model.

To counter this phenomenon, we developed \emph{phasing}, a technique to prefer ``small'' over ``large'' solutions, see Alg.~\ref{alg3}.
The idea of phasing is to alternate between a \emph{$\sat$-phase} and an \emph{$\unsat$-phase}.
In the $\sat$-phase, the goal is to find small candidate models.
To this end, we impose a bound on the absolute values of the exponents to block large models (Line~\ref{alg3:bound}).
If we fail to find a candidate model in the $\sat$-phase, then we switch to the $\unsat$-phase (Line~\ref{alg3:switch-unsat}).
In the $\unsat$-phase, our goal is prove unsatisfiability, so we do not impose any bounds on the exponents in Line~\ref{alg3:bound}.
If we fail to prove unsatisfiability, we discard the found model
  (unless it can be lifted to an \EIA-model) in Line~\ref{alg3:discard} and switch back to the $\sat$-phase (Line~\ref{alg3:switch-sat}), but with a higher bound (Line~\ref{alg3:inc}).

Note that the bound grows exponentially w.r.t.\ the number of phase changes (Lines~\ref{alg3:bound} and \ref{alg3:inc}).
Thus, if there are no small models, then the size of the considered models grows rapidly, so that we can still find a solution quickly.

\section{Related Work}
\label{sec:related}

The most closely related work applies \emph{incremental linearization} to \NIA, or to non-linear real arithmetic with transcendental functions (\NRAT).
Like our approach, incremental linearization is an instance of the CEGAR paradigm:
An initial abstraction (where certain predefined functions are considered as uninterpreted functions) is refined via linear lemmas that rule out the current counterexample.

Our approach is inspired by, but differs significantly from the approach for linearization of \NRAT by \textcite{cimatti18a}.
There, non-linear polynomials are linearized as well, whereas we leave the handling of polynomials to the backend solver.
Moreover, \citeauthor{cimatti18a} use linear lemmas only, whereas we also use bilinear lemmas.
Furthermore, \citeauthor{cimatti18a} fix the base to Euler's number $e$, whereas we consider a binary version of exponentiation.

The only lemmas of \textcite{cimatti18a}
that easily carry over to our approach are monotonicity lemmas.
While \citeauthor{cimatti18a} also use symmetry lemmas, they express properties of the sine function, i.e., they are fundamentally different from ours.
Our bounding lemmas are related to the ``lower bound'' and ``zero'' lemmas of \citeauthor{cimatti18a}, but there, $\lambda x.\ e^x$ is trivially bounded by $0$.
Interpolation lemmas are related to the ``tangent'' and ``secant lemmas'' of \citeauthor{cimatti18a}.
However, tangent lemmas make use of first derivatives, so they are not expressible with integer arithmetic in our setting, as we have $\frac{\deriv}{\deriv y} x^y = x^y \cdot \ln x$.
Secant lemmas are essentially obtained by linear interpolation, so our interpolation lemmas can be seen as a generalization of secant lemmas to binary functions.
A preprocessing by rewriting, prime lemmas, or induction lemmas are not considered by \citeauthor{cimatti18a}.

\textcite{cimatti18} also applied incremental linearization to \NIA.
Here, they used similar lemmas as for
NRAT \cite{cimatti18a}, so they differ again fundamentally from ours.

Further existing approaches for \NRAT are based on interval propagation \cite{isat3,dreal}.
As observed by \textcite{cimatti18a}, interval propagation effectively computes a piecewise \emph{constant} approximation, which is less expressive than our bilinear approximations.

In recent years, a novel approach for \NRAT based on the \emph{topological degree test} has been proposed \cite{cimatti22,LippariniR23}.
Its strength is finding irrational solutions more often than other approaches for \NRAT.
Hence, this line of work is orthogonal to ours.

\EIA could also be tackled by combining \NRAT techniques with branch-and-bound, but the following example shows that doing so is not promising.
\begin{example}
	\label{ex:related}
	Consider the formula $x = \exp(3, y) \land y > 0$.
	To handle it with existing solvers, we have to encode it using the natural exponential function:
	\begin{equation}
		\label{eq:e-enconding}
		e^z = 3 \land x = e^{y \cdot z} \land y > 0
	\end{equation}
	Here $x$ and $y$ range over the integers and $z$ ranges over the reals.
	Any model of \eqref{eq:e-enconding} satisfies $z = \ln 3$, where $\ln 3$ is irrational.
	As finding such models is challenging, the leading tools \tool{MathSat} \cite{mathsat}
	and \tool{CVC5} \cite{cvc5} fail for $e^z = 3$.
\end{example}
\tool{MetiTarski} \cite{metitarski} integrates decision procedures for real closed fields and approximations for transcendental functions into the theorem prover \tool{Metis} \cite{metis} to prove theorems about the reals.
In a related line of work, \tool{iSAT3} \cite{isat3} has been coupled with \tool{SPASS} \cite{spass}.
Clearly, these approaches differ fundamentally from ours.

Recently, the
complexity of decidable extensions of linear integer arithmetic with exponentiation has been investigated \cite{LIAE,icalp24,ouaknine25}.
\textcite{LIAE} consider \emph{Semenov arithmetic}, which extends LIA with a function $\lambda x. c^{|x|}$ for some fixed natural number $c > 1$.
Thus, it is equivalent to \EIA with quantifiers, but without the functions ``$\cdot$'',
``$\div$'', and ``$\mod$'', and where the first argument of all occurrences of $\exp$ must be the same constant.
In contrast, \textcite{icalp24} consider only quantifier-free conjunctions of literals from Semenov arithmetic.
Integrating such decision procedures into our approach is an interesting direction for future work.
The technique of \citeauthor{icalp24} is more efficient than the one of \citeauthor{LIAE}, but as it can only handle conjunctions of literals, it needs to be implemented as a theory solver, i.e., \emph{within} an SMT solver.
In contrast, our approach can be implemented \emph{on top of} an SMT solver, which is easier, as it does not require detailed knowledge of the internals of the SMT solver.
Thus, the approach of \citeauthor{LIAE} seems to be a better fit for our framework.

\textcite{ouaknine25} showed that decidability of the existential fragment of an extension of Semenov arithmetic with a second function $\lambda x. d^{|x|}$ (where $c$ and $d$ are multiplicatively independent) would lead to mathematical breakthroughs.
So as it stands, we cannot hope for a complete solver as soon as we go beyond Semenov arithmetic.

An alternative way to solve SMT problems with exponentiation is to provide the recursive definition of $\exp$ as universally quantified axioms.
Then SMT solvers with support for quantifiers try to instantiate these axioms suitably in order to derive a conflict.
To find suitable instantiations, various heuristics have been proposed \cite{ematching,conflict-based-instantiation,model-based-instantiation,enumerative-instantiation,counterexample-guided-instantiation}.
However, such approaches are only suitable for proving \emph{un}satisfiability, and for that purpose, our evaluation shows that our approach outperforms quantifier instantiation techniques, see \Cref{sec:purrs}.

\section{Implementation and Evaluation}
\label{sec:evaluation}

We implemented our approach in our novel tool \swinezzz, which is based on \tool{Z3} 4.14.1 \cite{z3}.
For more information on \swinezzz, including precompiled releases for all platforms, we refer to
\cite{swine}.

We compare \swinezzz with the default configuration of the original implementation of \textcite{conference}, called \swinel in the sequel.
In contrast to \swinez, \swinel does not support prime lemmas, induction lemmas, and phasing.
While \swinezzz directly uses the API of \tool{Z3}, 
\swinel is based on \tool{SMT-Switch} \cite{smt-switch}, a library that offers a unified interface for various SMT solvers.
By default, \swinel also uses \tool{Z3}.
For more information on \swinel, we refer to \cite{swine}.

\subsection{Benchmarks}

To evaluate our approach, we synthesized a large collection of \EIA problems from verification benchmarks for safety, termination, and complexity analysis.
More precisely, we ran our verification tool \tool{LoAT} \cite{loat} on the benchmarks for \emph{linear Constrained Horn Clauses (CHCs)} with \emph{linear integer arithmetic} from the \href{https://chc-comp.github.io}{\emph{CHC Competitions}}
2022 and 2023
as well as on the benchmarks for \emph{Termination} and \emph{Complexity of Integer Transition Systems} from the \href{https://termination-portal.org/wiki/TPDB}{\emph{Termination Problems Database (TPDB)}}, the benchmark set of the \href{https://termination-portal.org/wiki/Termination_Competition}{\emph{Termination and Complexity Competition}}~\cite{termcomp}, and extracted all SMT problems with exponentiation that \tool{LoAT} created while analyzing these benchmarks.
Afterwards, we removed duplicates.

The resulting benchmark set consists of 4627 SMT problems (called \emph{\tool{LoAT} Problems} below), which are available online \cite{swine-web}:
\begin{itemize}
	\item 669 problems that resulted from the benchmarks of the CHC Competition '22 (called \emph{CHC Comp '22 Problems} below)
	\item 158 problems that resulted from the benchmarks of the CHC Competition '23 (\emph{CHC Comp '23 Problems})
	\item 3146 problems that resulted from the complexity benchmarks of the TPDB (\emph{Complexity Problems})
	\item 654 problems that resulted from the termination benchmarks of the TPDB (\emph{Termination Problems})
\end{itemize}

Unfortunately, these sets do not contain challenging unsatisfiable benchmarks: As shown in
the evaluation of \textcite{conference}, the unsatisfiable instances can mostly be solved
without reasoning about exponentials at all.
More precisely, most of these benchmarks are not just unsatisfiable in \EIA (where the interpretation of $\exp$ is fixed), but they are also unsatisfiable in \NIA, where $\exp$ is an uninterpreted function.
To obtain challenging unsatisfiable examples, we generated one more set of benchmarks by validating the results of the recurrence solver \tool{PURRS} \cite{purrs}.
More precisely, we considered all recurrence relations from \tool{PURRS}'s test suite\footnote{\url{https://github.com/aprove-developers/LoAT-purrs/blob/master/tests/recurrences}} where both the recurrence relation and its expected solution can be expressed with polynomials and exponentials.
Consider such a recurrence relation
\begin{equation}
  \label{eq:eval-rec}
  f(n) = t
\end{equation}
where $f$ is the inductively defined function, i.e., the term $t$ contains subterms of the form $f(n-c)$, $c \in \NN_{>0}$.
Let $k$ be the maximal natural number such that $f(n-k)$ is a subterm of $t$.
As an example, consider the recurrence relation \eqref{eq:rec} in \Cref{ex:induction}
where $t$ is $2 \cdot f(n-1) + 2^n$ and thus, $k = 1$.
Moreover, let $s$ be the expected solution for \eqref{eq:eval-rec} given in the test suite.
For such expected solutions, we always have $n \in \VV(s) \subseteq \VV(t)$, and for all
subterms $f(q)$ of $s$, we have $q \in \{0,\ldots,k-1\}$.
Here, $\VV(s)$ and $\VV(t)$ denote the sets of variables occurring in $s$ and $t$,
respectively. For the solution \eqref{eq:sol} in \Cref{ex:induction}, we have
$s = (f(0) + n) \cdot 2^n$, i.e., here the only $f$-subterm of $s$ is $f(0)$. 
To construct new benchmarks, we considered
the following SMT problem:
\begin{equation}
  \label{eq:rec-smt}
  n \geq k \land s_\exp \neq t_\exp[f(n-i) / s_\exp[n/n-i] \mid 1 \leq i \leq k]
\end{equation}
where $s_\exp$ and $t_\exp$ result from $s$ and $t$ by replacing all exponentials (i.e.,
all subterms of the form $p^q$ where $q \notin \ZZ$) with $\exp(p,q)$.
Moreover, $s_\exp[n/n-i]$ results from $s_\exp$ by replacing $n$ with $n-i$ and
$t_\exp[f(n-i) / s_\exp[n/n-i] \mid 1 \leq i \leq k]$ results from $t_\exp$ by replacing
each subterm $f(n-i)$ with $s_\exp[n/n-i]$, for all $1 \leq i \leq k$.
In this way, we obtained the formula \eqref{recExBenchmark}  in  \Cref{ex:induction}.
If \eqref{eq:rec-smt} is unsatisfiable, then the given solution is correct.
More precisely, let $\vec{x} = \VV(s) \setminus \{n\}$.
Then unsatisfiability of \eqref{eq:rec-smt} implies that the equation \eqref{eq:eval-rec} holds for all $n \in \NN$, all $\vec{x} \in \ZZ^*$, and all $c_0,\ldots,c_{k-1} \in \ZZ$ when interpreting $f$ as
\[
  n \mapsto
  \begin{cases}
    c_n & \text{if } n < k \\
    s[f(i) / c_i \mid 0 \leq i < k] & \text{otherwise.}
  \end{cases}
\]
In this way, we obtained $37$ hard unsatisfiable benchmarks (called \emph{\tool{PURRS} Problems} below), which are also available online \cite{swine-web}.

\subsection{Evaluation on the \tool{LoAT} Problems}
\label{sec:loat}

We first report on our experiments on the \tool{LoAT} Problems.
To evaluate the impact of the different components of our approach, we tested with
configurations where we disabled rewriting, symmetry lemmas, monotonicity lemmas, bounding lemmas, prime lemmas, induction lemmas, interpolation lemmas, or
phasing.
All experiments were performed on the \href{https://help.itc.rwth-aachen.de/service/rhr4fjjutttf/article/fbd107191cf14c4b8307f44f545cf68a/}{CLAIX-2023-HPC nodes} of the RWTH Uni\-ver\-si\-ty High Performance Computing Cluster\footnote{\url{https://help.itc.rwth-aachen.de/service/rhr4fjjutttf}} with a memory limit of 10560 MiB ($\approx$~11GB) and a timeout of 10~s per example.
We chose a small timeout, as \tool{LoAT} usually has to discharge many SMT problems to solve a single verification task.
So in our setting, each individual SMT problem should be solved quickly.

\begin{figure}[p]
  \begin{minipage}{0.5\textwidth}
    \resizebox{\textwidth}{!}{%
      \begin{tabular}{|c|c|c|c|c|c|}
        \hline solver          & configuration            & $\sat$ & $\unsat$ & $\unknown$ \\
        \hline\hline \swinez   & \multirow{2}{*}{default} & 296    & 373      & 0          \\
        \hhline{-~----}
        \swinel                &                          & 296    & 373      & 0          \\
        \hline \hline \multirow{8}{*}{\swinez}
                               & no rewriting             & 296    & 373      & 0          \\
        \hhline{~-----}
                               & no symmetry              & 296    & 373      & 0          \\
        \hhline{~-----}
                               & no monotonicity          & 296    & 373      & 0          \\
        \hhline{~-----}
                               & no bounding              & 111    & 373      & 185        \\
        \hhline{~-----}
                               & no prime                 & 296    & 373      & 0          \\
        \hhline{~-----}
                               & no induction             & 296    & 373      & 0          \\
        \hhline{~-----}
                               & no interpolation         & 252    & 373      & 44         \\
        \hhline{~-----}
                               & no phasing               & 296    & 373      & 0          \\
        \hline
      \end{tabular}
    }
    \caption{CHC Comp '22, Results}
    \label{tab1}
  \end{minipage}
  \vspace{1em}
  \begin{minipage}{0.4\textwidth}
    \begin{tikzpicture}[scale=0.65]
      \begin{axis}[ legend pos=south east, xlabel=runtime in $\nicefrac{1}{10}$ seconds, ylabel=solved instances, ymin=500 ]
        \addplot[color=black,solid,thick] table[col sep=comma,header=false,x index=0,y index=1] {swine_z3_CHC_Comp_22_LIA_Lin.csv}; \addlegendentry{\swinez}
        \addplot[color=violet,densely dashed,thick] table[col sep=comma,header=false,x index=0,y index=1] {swine_CHC_Comp_22_LIA_Lin.csv}; \addlegendentry{\swinel}
      \end{axis}
    \end{tikzpicture}
    \caption{CHC Comp '22, Runtime}
    \label{rt1}
  \end{minipage}
  \begin{minipage}{0.5\textwidth}
    \resizebox{\textwidth}{!}{%
      \begin{tabular}{|c|c|c|c|c|c|}
        \hline solver          & configuration            & $\sat$ & $\unsat$ & $\unknown$ \\
        \hline\hline \swinez   & \multirow{2}{*}{default} & 87     & 71       & 0          \\
        \hhline{-~----}
        \swinel                &                          & 87     & 71       & 0          \\
        \hline \hline \multirow{8}{*}{\swinez}
                               & no rewriting             & 87     & 71       & 0          \\
        \hhline{~-----}
                               & no symmetry              & 87     & 71       & 0          \\
        \hhline{~-----}
                               & no monotonicity          & 87     & 71       & 0          \\
        \hhline{~-----}
                               & no bounding              & 79     & 71       & 8          \\
        \hhline{~-----}
                               & no prime                 & 87     & 71       & 0          \\
        \hhline{~-----}
                               & no induction             & 87     & 71       & 0          \\
        \hhline{~-----}
                               & no interpolation         & 38     & 71       & 49        \\
        \hhline{~-----}
                               & no phasing               & 87     & 71       & 0          \\
        \hline
      \end{tabular}
    }
    \caption{CHC Comp '23, Results}
    \label{tab2}
  \end{minipage}
  \vspace{1em}
  \begin{minipage}{0.4\textwidth}
    \begin{tikzpicture}[scale=0.65]
      \begin{axis}[ legend pos=south east, xlabel=runtime in $\nicefrac{1}{10}$ seconds, ylabel=solved instances, ymin=100 ]
        \addplot[color=black,solid,thick] table[col sep=comma,header=false,x index=0,y index=1] {swine_z3_CHC_Comp_23_LIA_Lin.csv}; \addlegendentry{\swinez}
        \addplot[color=violet,densely dashed,thick] table[col sep=comma,header=false,x index=0,y index=1] {swine_CHC_Comp_23_LIA_Lin.csv}; \addlegendentry{\swinel}
      \end{axis}
    \end{tikzpicture}
    \caption{CHC Comp '23, Runtime}
    \label{rt2}
  \end{minipage}
  \begin{minipage}{0.5\textwidth}
    \resizebox{\textwidth}{!}{%
      \begin{tabular}{|c|c|c|c|c|c|}
        \hline solver          & configuration            & $\sat$ & $\unsat$ & $\unknown$ \\
        \hline\hline \swinez   & \multirow{2}{*}{default} & 1355   & 1789     & 2          \\
        \hhline{-~----}
        \swinel                &                          & 1299   & 1789     & 58         \\
        \hline \hline \multirow{8}{*}{\swinez}
                               & no rewriting             & 1195   & 1789     & 162        \\
        \hhline{~-----}
                               & no symmetry              & 951    & 1789     & 406        \\
        \hhline{~-----}
                               & no monotonicity          & 1355   & 1789     & 2          \\
        \hhline{~-----}
                               & no bounding              & 890    & 1789     & 467        \\
        \hhline{~-----}
                               & no prime                 & 1355   & 1789     & 2          \\
        \hhline{~-----}
                               & no induction             & 1355   & 1789     & 2          \\
        \hhline{~-----}
                               & no interpolation         & 1241   & 1788     & 117        \\
        \hhline{~-----}
                               & no phasing               & 1271   & 1789     & 86         \\
        \hline
      \end{tabular}
    }
    \caption{Complexity, Results}
    \label{tab3}
  \end{minipage}
  \vspace{1em}
  \begin{minipage}{0.4\textwidth}
    \begin{tikzpicture}[scale=0.65]
      \begin{axis}[ legend pos=south east, xlabel=runtime in $\nicefrac{1}{10}$ seconds, ylabel=solved instances, ymin=2000 ]
        \addplot[color=black,solid,thick] table[col sep=comma,header=false,x index=0,y index=1] {swine_z3_TPDB_ITS_Complexity.csv}; \addlegendentry{\swinez}
        \addplot[color=violet,densely dashed,thick] table[col sep=comma,header=false,x index=0,y index=1] {swine_TPDB_ITS_Complexity.csv}; \addlegendentry{\swinel}
      \end{axis}
    \end{tikzpicture}
    \caption{Complexity, Runtime}
    \label{rt3}
  \end{minipage}
  \begin{minipage}{.5\textwidth}
    \resizebox{\textwidth}{!}{%
      \begin{tabular}{|c|c|c|c|c|c|}
        \hline solver          & configuration            & $\sat$ & $\unsat$ & $\unknown$ \\
        \hline\hline \swinez   & \multirow{2}{*}{default} & 223    & 431      & 0          \\
        \hhline{-~----}
        \swinel                &                          & 223    & 431      & 0          \\
        \hline \hline \multirow{8}{*}{\swinez}
                               & no rewriting             & 223    & 431      & 0          \\
        \hhline{~-----}
                               & no symmetry              & 223    & 431      & 0          \\
        \hhline{~-----}
                               & no monotonicity          & 223    & 431      & 0          \\
        \hhline{~-----}
                               & no bounding              & 190    & 429      & 35         \\
        \hhline{~-----}
                               & no prime                 & 223    & 431      & 0          \\
        \hhline{~-----}
                               & no induction             & 223    & 431      & 0          \\
        \hhline{~-----}
                               & no interpolation         & 155    & 429      & 70         \\
        \hhline{~-----}
                               & no phasing               & 223    & 431      & 0          \\
        \hline
      \end{tabular}
    }
    \caption{Termination, Results}
    \label{tab4}
  \end{minipage}
  \begin{minipage}{0.4\textwidth}
    \begin{tikzpicture}[scale=0.65]
      \begin{axis}[ legend pos=south east, xlabel=runtime in $\nicefrac{1}{10}$ seconds, ylabel=solved instances, ymin=500 ]
        \addplot[color=black,solid,thick] table[col sep=comma,header=false,x index=0,y index=1] {swine_z3_TPDB_ITS_Termination.csv}; \addlegendentry{\swinez}
        \addplot[color=violet,densely dashed,thick] table[col sep=comma,header=false,x index=0,y index=1] {swine_TPDB_ITS_Termination.csv}; \addlegendentry{\swinel}
      \end{axis}
    \end{tikzpicture}
    \caption{Termination, Runtime}
    \label{rt4}
  \end{minipage}
\end{figure}

The results can be seen in \Cref{tab1,tab2,tab3,tab4}.
All but two of the $4627$ benchmarks can be solved by \swinez.
Most configurations, including \swinel, can solve all CHC Comp and Termination Problems.
The only exceptions are the configurations without bounding or interpolation lemmas, which clearly shows their importance.

The Complexity Problems (\Cref{tab3}) are by far the hardest, and thus also the most interesting.
Here, \swinez clearly outperforms \swinel (2 vs.\ 58 unsolved instances).
Moreover, \Cref{tab3} shows that all components of our approach except for monotonicity
lemmas, prime lemmas, and induction lemmas are crucial for the performance of \swinez on
this benchmark set.
Thus, \emph{phasing} is
the main reason why \swinez is superior to \swinel on the satisfiable
instances from our benchmark suite.

The fact that monotonicity lemmas are of little importance here is not surprising:
Monotonicity lemmas were already of little use when \tool{Z3} was used as backend solver
in the evaluation of \swinel of
\textcite{conference}.
However, these lemmas significantly improved the performance of \swinel with \tool{CVC5} \cite{cvc5} as backend solver.
Hence, the usefulness of monotonicity lemmas depends on the details of the underlying solver.
Moreover, recall that our leading example cannot be solved without monotonicity lemmas, see \Cref{Leading Example -- Total Semantics}.

It is also not surprising that induction lemmas do not contribute much to the performance
of \swinez on this benchmark set, as their main purpose is proving unsatisfiability, see \Cref{sec:purrs}.

However, it is slightly disappointing that prime lemmas are not needed for the \tool{LoAT} problems.
Nevertheless, we believe that prime factorizations are of such fundamental relevance that
prime lemmas are likely to be useful for benchmarks from other sources in the future.

The runtime of \swine can be seen in \Cref{rt1,rt2,rt3,rt4}.
Most instances can be solved in a fraction of a second, as desired for our use case.
On the CHC Comp '23 and the Termination Problems, \swinez and \swinel are equally efficient.
On the CHC Comp '22 and the Complexity Problems, \swinez is slightly faster than
\swinel.
We refer to \textcite{swine-web} for more details on our evaluation.

\subsection{Evaluation on the \tool{PURRS} Problems}
\label{sec:purrs}

We continue with the evaluation on the \tool{PURRS} Problems, where we used a wall clock timeout of $300$s.
All of them are unsatisfiable (provided that the test cases from \tool{PURRS}'s test suite are correct), and most of these problems require inductive reasoning, so they are particularly suitable to illustrate the usefulness of induction lemmas.
Hence, we considered two configurations of \swinez: The default configuration (where all features are enabled), and a configuration without induction lemmas.

In contrast to those \tool{LoAT} Problems that are unsatisfiable,
the \tool{PURRS} Problems require reasoning about $\exp$ for proving
unsatisfiability.
An alternative approach to tackle such problems is to use standard SMT solvers, and to add the (universally quantified) recursive definition of $\exp$ explicitly.
Thus, we also considered \tool{Z3} 4.14.1, where we added the following assertions to each benchmark:
\begin{align*}
  \forall x \in \ZZ.\ & \exp(x, 0) = 1 \\
  \forall x, y \in \ZZ.\ & y > 0 \implies \exp(x,y) = x \cdot \exp(x, y-1) \\
  \forall x, y \in \ZZ.\ & y < 0 \implies \exp(x,y) = x \cdot \exp(x, y+1)
\end{align*}
Note that SMT solvers use quantifier instantiation to derive contradictions, but adding such quantified assertions prevents them from proving satisfiability.
Hence, we did not include \tool{Z3} in the evaluation on the \tool{LoAT} problems in \Cref{sec:loat}.

\begin{table}
  \begin{tabular}{|c||c|c||c|c||c|c|}
    \hline \textbf{tool} & \multicolumn{2}{c||}{\swinez} & \multicolumn{2}{c||}{\tool{Z3}} & \multicolumn{2}{c|}{\swinez} \\\hline
    \textbf{configuration} & \multicolumn{2}{c||}{default} & \multicolumn{2}{c||}{default} & \multicolumn{2}{c|}{no induction} \\\hline\hline
    \texttt{purrs01} & \checkmark & 0.45 & \checkmark & 0.66 & \clock & -- \\\hline
    \texttt{purrs02} & \checkmark & 0.40 & \checkmark & 0.66 & \clock & -- \\\hline
    \texttt{purrs03} & \clock & -- & \clock & -- & \clock & -- \\\hline
    \texttt{purrs04} & \checkmark & 0.44 & \checkmark & 0.65 & \clock & -- \\\hline
    \texttt{purrs05} & \checkmark & 0.43 & \checkmark & 0.44 & \clock & -- \\\hline
    \texttt{purrs06} & \checkmark & 0.45 & \checkmark & 0.45 & \clock & -- \\\hline
    \texttt{purrs07} & \checkmark & 0.60 & \checkmark & 31.64 & \xmark & 26.7 \\\hline
    \texttt{purrs08} & \checkmark & 0.58 & \checkmark & 0.46 & \checkmark & 0.49 \\\hline
    \texttt{purrs09} & \checkmark & 0.52 & \xmark & 104.58 & \clock & -- \\\hline
    \texttt{purrs10} & \checkmark & 0.50 & \xmark & 108.85 & \clock & -- \\\hline
    \texttt{purrs11} & \checkmark & 0.48 & \checkmark & 0.49 & \clock & -- \\\hline
    \texttt{purrs12} & \checkmark & 3.12 & \checkmark & 2.17 & \clock & -- \\\hline
    \texttt{purrs13} & \checkmark & 0.49 & \checkmark & 0.61 & \clock & -- \\\hline
    \texttt{purrs14} & \checkmark & 0.52 & \checkmark & 0.52 & \checkmark & 0.42 \\\hline
    \texttt{purrs15} & \checkmark & 0.43 & \clock & -- & \clock & -- \\\hline
    \texttt{purrs16} & \checkmark & 1.35 & \clock & -- & \clock & -- \\\hline
    \texttt{purrs17} & \checkmark & 1.42 & \clock & -- & \clock & -- \\\hline
    \texttt{purrs18} & \checkmark & 0.54 & \checkmark & 0.48 & \clock & -- \\\hline
    \texttt{purrs19} & \checkmark & 0.52 & \checkmark & 0.48 & \clock & -- \\\hline
    \texttt{purrs20} & \checkmark & 0.53 & \checkmark & 0.49 & \clock & -- \\\hline
    \texttt{purrs21} & \checkmark & 0.41 & \checkmark & 0.43 & \clock & -- \\\hline
    \texttt{purrs22} & \checkmark & 0.43 & \checkmark & 0.46 & \clock & -- \\\hline
    \texttt{purrs23} & \checkmark & 0.47 & \checkmark & 0.44 & \clock & -- \\\hline
    \texttt{purrs24} & \checkmark & 0.46 & \checkmark & 0.45 & \clock & -- \\\hline
    \texttt{purrs25} & \checkmark & 0.43 & \checkmark & 0.44 & \clock & -- \\\hline
    \texttt{purrs26} & \checkmark & 0.43 & \checkmark & 0.42 & \clock & -- \\\hline
    \texttt{purrs27} & \checkmark & 0.40 & \checkmark & 0.46 & \clock & -- \\\hline
    \texttt{purrs28} & \checkmark & 0.44 & \checkmark & 0.51 & \clock & -- \\\hline
    \texttt{purrs29} & \checkmark & 0.44 & \checkmark & 0.53 & \xmark & 239.34 \\\hline
    \texttt{purrs30} & \checkmark & 2.21 & \clock & -- & \clock & -- \\\hline
    \texttt{purrs31} & \xmark & 18.93 & \clock & -- & \xmark & 87.65 \\\hline
    \texttt{purrs32} & \xmark & 77.34 & \clock & -- & \clock & -- \\\hline
    \texttt{purrs33} & \checkmark & 0.43 & \checkmark & 0.46 & \clock & -- \\\hline
    \texttt{purrs34} & \checkmark & 0.50 & \checkmark & 0.51 & \checkmark & 0.38 \\\hline
    \texttt{purrs35} & \checkmark & 21.39 & \checkmark & 6.12 & \clock & -- \\\hline
    \texttt{purrs36} & \checkmark & 4.65 & \checkmark & 38.30 & \clock & -- \\\hline
    \texttt{purrs37} & \checkmark & 1.36 & \checkmark & 1.11 & \clock & -- \\\hline
    \hline
    \textbf{solved} & \multicolumn{2}{c||}{34} & \multicolumn{2}{c||}{28} & \multicolumn{2}{c|}{3} \\\hline
    \multicolumn{7}{c}{}\\[-.5em]
    \multicolumn{7}{c}{\checkmark: \unsat, \xmark: \unknown, \clock: timeout}
  \end{tabular}
  \caption{\tool{PURRS} Problems}
  \label{tab:purrs}
\end{table}

The results can be seen in \Cref{tab:purrs}.
With induction lemmas, \swinez solves all instances that \tool{Z3} can solve as well as 6 additional instances, and it only fails for 3 instances.
Regarding the benchmarks that are solved by both tools, \swinez is
substantially
faster in some cases (\texttt{purrs07}, \texttt{purrs36}), but sometimes it is the other way around (\texttt{purrs35}), so there is no clear picture.
Without induction lemmas, \swinez is clearly not suitable for this benchmark
set.

\section{Conclusion}

We presented the novel SMT theory \EIA, which extends the theory \emph{non-linear integer arithmetic} with integer exponentiation.
Moreover, inspired by \emph{incremental linearization} for similar extensions of \emph{non-linear real arithmetic}, we developed a CEGAR approach to solve \EIA problems.
The core idea of our approach is to regard exponentiation as an uninterpreted function and to eliminate counterexamples, i.e., models that violate the semantics of exponentiation, by generating suitable \emph{lemmas}.
Here, the use of \emph{bilinear interpolation} turned out to be crucial for proving satisfiability, both in practice (see our evaluation in \Cref{sec:evaluation}) and in theory, as interpolation lemmas are essential for being able to eliminate \emph{any} counterexample (see \Cref{ProgressTheorem}).
Finally, we evaluated the implementation of our approach in our novel tool \swinez on thousands of \EIA problems that were synthesized from verification tasks using our verification tool \tool{LoAT}, and on challenging unsatisfiable benchmarks that result from the verification of solutions for recurrence relations.
Our evaluation shows that \swinez is highly effective for both use cases.
In particular, it shows that \emph{induction lemmas} are crucial for proving
unsatisfiability, and that our new \emph{phasing} technique allows our novel
implementation \swinez to outperform its predecessor on satisfiable instances. 

Consequently, we integrated \swinez into our verification tool \tool{LoAT}, where it now serves as the default solver for all techniques that require reasoning about exponentials \cite{adcl,abmc}.

With \swinez, we provide an SMT-LIB compliant open-source solver for \EIA \cite{swine}.
In this way, we hope to attract users with applications that give rise to challenging benchmarks, and we hope that other solvers with support for integer exponentiation will follow, with the ultimate goal of standardizing \EIA.

\begin{acks}
  This work has been funded by the Deutsche Forschungsgemeinschaft (DFG, German Research Foundation) - 235950644 (Project GI 274/6-2).
\end{acks}

\clearpage

\printbibliography

\begin{appendix}
  \begin{appendix}
\appendixproofsection{Missing Proofs}\label{sec:MissingProofs}
\appendixproof*{lem:preproc}
\appendixproof*{lem:symmetry}
\appendixproof*{lem:bounding}
% \appendixproof*{lem:interpol}
% \appendixproof*{lem:SoundnessInterpolationUpper}
% \appendixproof*{lem:SoundnessInterpolationLower}
% \appendixproof*{lem:monotonicity}
% \appendixproof*{ProgressTheorem}
\end{appendix}

\end{appendix}

\end{document}